\documentclass[10pt]{article}
\usepackage{times,color,graphicx}               
\usepackage{amsmath,amssymb,amsthm,url}
\newcommand{\deriv}[3][]{\frac{\d^{#1}{#2}}{{\d{#3}}^{#1}}}
\newcommand{\pderiv}[3][]{\frac{\partial^{#1}{#2}}{{\partial{#3}}^{#1}}}

\def\O{\mathcal{O}}
\def\p{Pain\-lev\'e}
\def\bk{B\"ack\-lund}

\def\bts{B\"ack\-lund transformations}
\newcommand{\comment}[1]{}

\def\d{{\rm d}}

\def\i{\ifmmode{\rm i}\else\char"10\fi}

\newcommand{\Integer}{\mathbb{Z}}

\def\Z{\Integer}

\def\ode{ordinary differential equation}
\def\odes{ordinary differential equations}

\def\bts{B\"ack\-lund transformations}

\def\peq{\p\ equation}
\def\peqs{\p\ equations}
\def\PI{\mbox{\rm P$_{\rm I}$}}
\def\PII{\mbox{\rm P$_{\rm II}$}}
\def\PIII{\mbox{\rm P$_{\rm III}$}}
\def\PIV{\mbox{\rm P$_{\rm IV}$}}
\def\PV{\mbox{\rm P$_{\rm V}$}}
\def\PVI{\mbox{\rm P$_{\rm VI}$}}
\def\HII{\mathcal{H}_{\rm II}} 
\def\SII{\mbox{\rm S$_{\rm II}$}}
\def\Ptf{$\mbox{\rm P}_{34}$}
\def\dPI{\mbox{\rm dP$_{\rm I}$}}

\def\det{\mathop{\rm det}\nolimits}

\def\Tr{\mathop{\rm Tr}\nolimits}
\def\ds{\displaystyle}

\def\Ai{\mathop{\rm Ai}\nolimits}
\def\Bi{\mathop{\rm Bi}\nolimits}

\newcommand{\BesselJ}[1]{J_{#1}}
\newcommand{\BesselY}[1]{Y_{#1}}

\newcommand{\WhitD}[1]{D_{#1}}
\newcommand{\WhitM}[2]{M_{#1,#2}}
\newcommand{\WhitW}[2]{q_{#1,#2}}
\newcommand{\KummerM}{M}
\newcommand{\KummerU}{U}

\newcommand{\HyperpFq}[2]{{}_{#1}F_{#2}}
\newcommand{\JacobiP}[3]{P^{(#1,#2)}_{#3}}
\newcommand{\HermiteH}[1]{H_{#1}}

\newcommand{\LaguerreL}[2]{L^{(#1)}_{#2}}

\def\a{\alpha}
\def\b{\beta}

\def\k{\kappa}

\def\th{\vartheta}

\def\ep{\varepsilon}

\newtheorem{theorem}{Theorem}[section]

\newtheorem{lemma}[theorem]{Lemma}

\theoremstyle{definition}

\newtheorem{remark}[theorem]{Remark}

\newtheorem{conjecture}[theorem]{Conjecture}

\numberwithin{figure}{section}
\numberwithin{equation}{section}
\numberwithin{table}{section}

\def\sref#1{\S\ref{#1}}

\definecolor{dkg}{rgb}{0,0.5,0}
\definecolor{purple}{rgb}{0.5,0,0.7}
\def\blue#1{\textcolor{blue}{#1}}
\def\green#1{\textcolor{dkg}{#1}}
\def\purple#1{\textcolor{purple}{#1}}
\def\red#1{\textcolor{red}{#1}}

\def\fig#1{\includegraphics[width=5cm]{#1}}
\def\figg#1#2{\includegraphics[width=#1]{#2}}
\begin{document}

\title{On Airy Solutions of the Second Painlev\'e Equation}
\author{Peter A.\ Clarkson\\ School of Mathematics, Statistics and Actuarial Science,\\
University of Kent, Canterbury, CT2 7NF, UK\\ 
\texttt{P.A.Clarkson@kent.ac.uk}}

\maketitle

\begin{abstract}
In this paper we discuss Airy  solutions of the second Painlev\'{e} equation (\PII) and two related equations, the Painlev\'{e} XXXIV equation (\Ptf) and the Jimbo-Miwa-Okamoto $\sigma$ form of \PII\ (\SII), are discussed. It is shown that solutions which depend only on the Airy function $\Ai(z)$ have a completely difference structure to those which involve a linear combination of the Airy functions $\Ai(z)$ and $\Bi(z)$. For all three equations, the special solutions which depend only on $\Ai(t)$ are \textit{tronqu\'{e}e} solutions, i.e.\ they have no poles in a sector of the complex plane. Further for  both \Ptf\ and \SII, it is shown that amongst these \textit{tronqu\'{e}e} solutions there is a family of solutions which have no poles on the real axis.
\end{abstract}

\begin{center}\textit{Dedicated to Mark Ablowitz on his 70th birthday}\end{center}

\section{Introduction}
The six \peqs\ (\PI--\PVI) were first discovered by \p, Gambier and their colleagues in an investigation of which second order \odes\ of the form 
\begin{equation} \label{eq:PT.INT.gen-ode} \deriv[2]{q}{z}=F\left(\deriv{q}{z},q,z\right),  \end{equation} 
where $F$ is rational in $dq/dz$ and $q$ and analytic in $z$, have the property that their solutions have no movable branch points. They showed that there were fifty canonical equations of the form (\ref{eq:PT.INT.gen-ode}) with this property, now known as the \textit{\p\ property}. Further \p, Gambier and their colleagues showed that of these fifty equations, forty-four can be reduced to linear equations, solved in terms of elliptic functions, or are reducible to one of six new nonlinear \odes\ that define new transcendental functions, see Ince \cite{refInce}. The \p\ equations can be thought of as nonlinear analogues of the classical special functions \cite{refPAC05review,refFIKN,refGLS02,refIKSY,refUmemura98}, and arise in a wide variety of applications, for example random matrices, cf.~\cite{refForrester,refOsipovKanz}.

In this paper we are concerned with special solutions of the second \peq\ (\PII)
\begin{equation}\label{eq:PT.DE.PII}
\deriv[2]{q}{z} = 2q^3 + z q + \a,
\end{equation}
with 
$\a$ an arbitrary constant, and two related equations. These are the \p\ XXXIV equation (\Ptf)
\begin{equation}\label{eq:PT.DE.P34}
\deriv[2]{p}{z} =\frac1{2p}\left(\deriv{p}{z}\right)^{\!2} 
+ 2p^2-zp-\frac{(\a+\tfrac12)^2}{2p},
\end{equation} 
is equivalent to equation XXXIV of Chapter 14 in \cite{refInce}, which is solvable in terms of \PII\ (\ref{eq:PT.DE.PII}), see \S\ref{ssec:Ham}, and  the 
Jimbo-Miwa-Okamoto $\sigma$ form of \PII\ (\SII)
\begin{equation}\label{eq:PT.DE.SII}
\left(\deriv[2]{\sigma}{z}\right)^{\!2}+4\left(\deriv{\sigma}{z}\right)^{\!3}+2\deriv{\sigma}{z}\left(z\deriv{\sigma}{z}-\sigma\right)
=\tfrac14(\a+\tfrac12)^2, \end{equation} 
which is satisfied by the Hamiltonian associated with \PII\ (\ref{eq:PT.DE.PII}), see \S\ref{ssec:Ham}. 
Equation (\ref{eq:PT.DE.SII}) is equation SD-I.d in the classification of second-order, second-degree equations which have the \p\ property by Cosgrove and Scoufis \cite{refCS}, an equation first derived by Chazy \cite{refChazy11}.
Frequently in applications it is the associated second-order, second-degree equation such as (\ref{eq:PT.DE.SII}) which arises rather than the \peq.

It is well-known that \PII\ has special solutions depending on one parameter that are expressed in terms of the Airy functions $\Ai(z)$ and $\Bi(z)$. In this paper we study the Airy solutions for \PII\ (\ref{eq:PT.DE.PII}), \Ptf\ (\ref{eq:PT.DE.P34}) and \SII\ (\ref{eq:PT.DE.SII}), see \S\ref{sec:sf}. In particular it is shown that the solutions which depend only on $\Ai(z)$ have a completely different asymptotic behaviour as $z\to-\infty$ in comparison to those which involve a linear combination of $\Ai(z)$ and $\Bi(z)$, which is a new characterization of these special solutions of \PII\ (\ref{eq:PT.DE.PII}), \Ptf\ (\ref{eq:PT.DE.P34}) and \SII\ (\ref{eq:PT.DE.SII}). Further, the special solutions of the three equations which depend only on $\Ai(t)$ are \textit{tronqu\'{e}e} solutions, i.e.\ they have no poles in a sector of the complex plane. Additionally it is shown that there are new families of Airy solutions of  \Ptf\ (\ref{eq:PT.DE.P34}) and \SII\ (\ref{eq:PT.DE.SII}) which have no poles on the real axis and so are likely to arise in physical applications. 
For \Ptf\ (\ref{eq:PT.DE.P34}) these  solutions have algebraic decay as $z\to\pm\infty$, whilst for \SII\ (\ref{eq:PT.DE.SII}) these solutions have algebraic decay as $z\to +\infty$ and algebraic growth as $z\to +\infty$.

\section{Some properties of the second \p\ equation}
\subsection{Hamiltonian structure}\label{ssec:Ham} Each of the
\peqs\ \PI--\PVI\ can be written as a Hamiltonian system 
\begin{equation}\label{sec:PT.HM.DE1} 
\deriv{q}{z}=\pderiv{\mathcal{H}_{\rm J}}{p},\qquad  
\deriv{p}{z}=-\pderiv{\mathcal{H}_{\rm J}}{q},
\end{equation}
for a suitable Hamiltonian function $\mathcal{H}_{\rm J}(q,p,z)$\ \cite{refJMi,refOkamoto80a,refOkamoto80b,refOkamotoPIIPIV}. 
The function $\sigma(z)\equiv\mathcal{H}_{\rm J}(q,p,z)$ satisfies a second-order, second-degree \ode, whose solution is expressible in terms of the solution of the associated \peq\ \cite{refJMi,refOkamoto80a,refOkamoto80b,refOkamotoPIIPIV}.

The Hamiltonian associated with \PII\ (\ref{eq:PT.DE.PII}) is
\begin{equation}\label{sec:PT.HM.DE4}
\HII(q,p,z;\a) = \tfrac12 p^2 - (q^2+\tfrac12 z)p - (\a+\tfrac12)q
\end{equation} and so
\begin{equation}\label{sec:PT.HM.DE3}
\deriv{q}{z}=p-q^2-\tfrac12z,\qquad \deriv{p}{z}=2qp+\a+\tfrac12
\end{equation} (Jimbo and Miwa \cite{refJMi}, Okamoto \cite{refOkamotoPIIPIV}). 
Eliminating $p$ in (\ref{sec:PT.HM.DE3}) then $q$ satisfies \PII\ (\ref{eq:PT.DE.PII}) whilst eliminating $q$ yields (\ref{eq:PT.DE.P34}).
Further if $q$ satisfies \PII\ (\ref{eq:PT.DE.PII}) then $p=q'+q^2+\tfrac12 z$ satisfies (\ref{eq:PT.DE.P34}). 
Conversely if $p$ satisfies (\ref{eq:PT.DE.P34}) then
$q=(p'-\a-\tfrac12)/(2p)$ satisfies \PII\ (\ref{eq:PT.DE.PII}). Thus there is a one-to-one
correspondence between solutions of \PII\ (\ref{eq:PT.DE.PII}) and those of \Ptf\
(\ref{eq:PT.DE.P34}). 

An important property of the Hamiltonian (\ref{sec:PT.HM.DE4}), which is very useful in applications, is that it satisfies a second-order, second-degree \ode, as discussed in the following theorem.

\begin{theorem}{\label{thm:2.2}Consider the function 
$\sigma(z;\a)=\HII(q,p,z;\a)$ defined by (\ref{sec:PT.HM.DE4}), where $q$ and $p$ satisfy the system (\ref{sec:PT.HM.DE3}), then
$\sigma(z;\a)$ satisfies (\ref{eq:PT.DE.SII}).
Conversely if $\sigma(z;\a)$ is a solution of (\ref{eq:PT.DE.SII}), then 
\begin{equation}\label{sec:PT.HM.DE6}
q(z;\a)=\frac{4\sigma''(z;\a)+2\a+1}{8\sigma'(z;\a)},
\qquad p(z;\a)=-2\sigma'(z;\a),
\end{equation} 
with $'\equiv d/dz$, are solutions of (\ref{eq:PT.DE.PII}) and (\ref{eq:PT.DE.P34}), respectively.}\end{theorem}

\begin{proof}See Jimbo and Miwa \cite{refJMi} and Okamoto \cite{refOkamoto80a,refOkamoto80b,refOkamotoPIIPIV}.\end{proof}

\subsection{\bk\ transformations}\label{sec:bts}
The \peqs\ \PII--\PVI\ possess \emph{\bts}\ which relate one solution to another solution either of the same equation, with different values of the parameters, or another equation (see \cite{refPAC05review,refGLS02,refPACDLMF,refFA82} and the references therein). An important application of the \bts\ is that they generate hierarchies of classical solutions of the \peqs, which are discussed in \sref{sec:sf}.

The \bts\ for \PII\ (\ref{eq:PT.DE.PII}) are given in the following theorem.
\begin{theorem}{Let $q\equiv q(z;\a)$ is a solution of \PII\ (\ref{eq:PT.DE.PII}), then the transformations
\begin{align}\label{eq:PT.BT.P21}
\mathcal{S}:&\qquad q(z;-\a)=-\,q,\\ 
\label{eq:PT.BT.P22}
\mathcal{T}_{\pm}:&\qquad q(z;\a\pm1)=-\,q-\frac{2\a\pm1}{2q^2\pm2q'+z},
\end{align} give solutions of \PII, provided that
$\a\not=\mp\tfrac12$ in (\ref{eq:PT.BT.P22}).}\end{theorem} 

\begin{proof}See Gambier \cite{refGambier09}.\end{proof}


The solutions $q_{\a}=q(z;\a)$, $q_{\a\pm1}=q(z;\a\pm1)$ also satisfy the nonlinear recurrence relation
\begin{equation}\label{eq:PT.BT.P23}
\frac{2\a+1}{q_{\a+1}+q_{\a}}
+\frac{2\a-1}{q_{\a}+q_{\a-1}} +4q_{\a}^2 + 2z=0,
\end{equation} 
a difference equation which is known as an alternative form of discrete \PI\ (alt-\dPI)\ \cite{refFGR}. The difference equation (\ref{eq:PT.BT.P23}) is obtained by eliminating $q'$ between the transformations $\mathcal{T}_{\pm}$ given by (\ref{eq:PT.BT.P22}). Note that for \PII\ (\ref{eq:PT.DE.PII}), the independent variable $z$ varies and the parameter $\a$ is fixed, whilst for the discrete equation (\ref{eq:PT.BT.P23}), $z$ is a fixed parameter and $\a$ varies.

\section{Special function solutions}\label{sec:sf} 
The \peqs\ \PII--\PVI\ possess hierarchies of solutions expressible in terms of classical special functions, for special values of the parameters through an associated Riccati equation, 
\begin{equation}
\label{eq:PT.SF.eq20} \deriv{q}{z} 
=f_2(z)q^2+f_1(z)q+f_0(z),
\end{equation} 
where $f_2(z)$, $f_1(z)$ and $f_0(z)$ are rational functions. Hierarchies of solutions, which are often referred to as ``one-parameter solutions" (since they have one arbitrary constant), are generated from ``seed solutions'' derived from the Riccati equation using the \bts\ given in \sref{sec:bts}. Furthermore, as for the rational solutions, these special function solutions are often expressed in the form of determinants. 

Solutions of \PII--\PVI\ are expressed in terms of special functions as follows (see \cite{refPAC05review,refGLS02,refPACDLMF,refMasuda04}, and the references therein):
for \PII\ in terms of Airy functions $\Ai(z)$, $\Bi(z)$;
for \PIII\ in terms of Bessel functions $\BesselJ{\nu}(z)$, $\BesselY{\nu}(z)$; 
for \PIV\ in terms of parabolic cylinder functions $\WhitD{\nu}(z)$; 
for \PV\ in terms of confluent hypergeometric functions $\HyperpFq11(a;c;z)$, equivalently Kummer functions $\KummerM(a,b,z)$, $\KummerU(a,b,z)$ or Whittaker functions $\WhitM{\k}{\mu}(z)$, $\WhitW{\k}{\mu}(z)$; 
and  for \PVI\ in terms of hypergeometric functions $\HyperpFq21(a,b;c;z)$. 
Some classical orthogonal polynomials arise as particular cases of these special functions and thus yield rational solutions of the associated \peqs:  
for \PIII\ and \PV\ in terms of associated Laguerre polynomials $\LaguerreL{m}{k}(z)$; 
for \PIV\ in terms of Hermite polynomials $\HermiteH{n}(z)$; and 
for \PVI\ in terms of Jacobi polynomials $\JacobiP{\a}{\b}{n}(z)$. 

\subsection{Second \peq}
We note that \PII\ (\ref{eq:PT.DE.PII}) can be written as
\[\ep \deriv{}{z}\left(\ep \deriv{q}{z}-q^2-\tfrac12z\right) +2q\left(\ep \deriv{q}{z}-q^2-\tfrac12z\right)=\a-\tfrac12\ep,\]
with $\ep^2=1$. Hence if $\a=\tfrac12\ep$, then special solutions of
\PII\ can be obtained in terms of solutions of the Riccati equation 
\begin{equation}\label{eq:PT.OP.eq21}
\ep \deriv{q}{z}=q^2+\tfrac12z.
\end{equation} Any solution of this equation is also a solution of \PII\ (\ref{eq:PT.DE.PII}),
provided that $\a=\tfrac12\ep$. Linearising the Riccati equation (\ref{eq:PT.OP.eq21}) by setting $\ds q=-{\ep}\varphi'/\varphi$
yields
\begin{equation}\label{eq:PT.OP.eq22}
\deriv[2]{\varphi}{z}+\tfrac12z\varphi=0,
\end{equation} which is equivalent to the Airy equation and has general
solution
\begin{equation}\label{eq:PT.OP.eq23}
\varphi(z;\th)=\cos(\th)\Ai(t)+\sin(\th)\Bi(t),\qquad t=-2^{-1/3}\,z,
\end{equation} with $\Ai(t)$ and $\Bi(t)$ the Airy functions and $\th$ an arbitrary constant.
The Airy solutions of \PII\ are classified in the following theorem due to Gambier \cite{refGambier09}.
\begin{theorem}{The second \peq\ (\ref{eq:PT.DE.PII}) has a one-parameter family of
solutions expressible in terms of Airy functions given by
(\ref{eq:PT.OP.eq23}) if and only if $\a=n-\tfrac12$, with $n\in\Integer$.}\end{theorem} 


\begin{figure}{\[ \begin{array}{c@{\quad}c@{\quad}c}
\fig{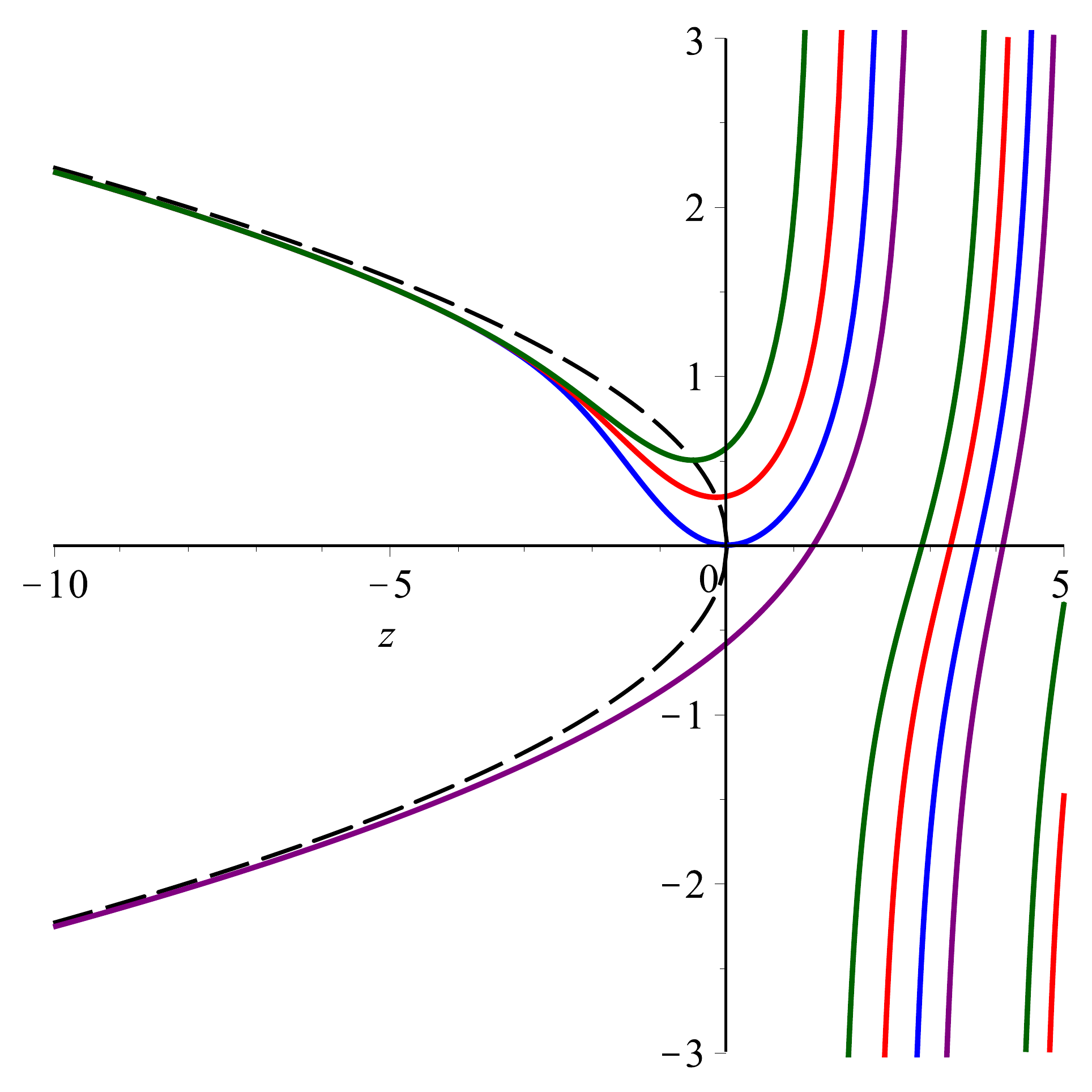} &\fig{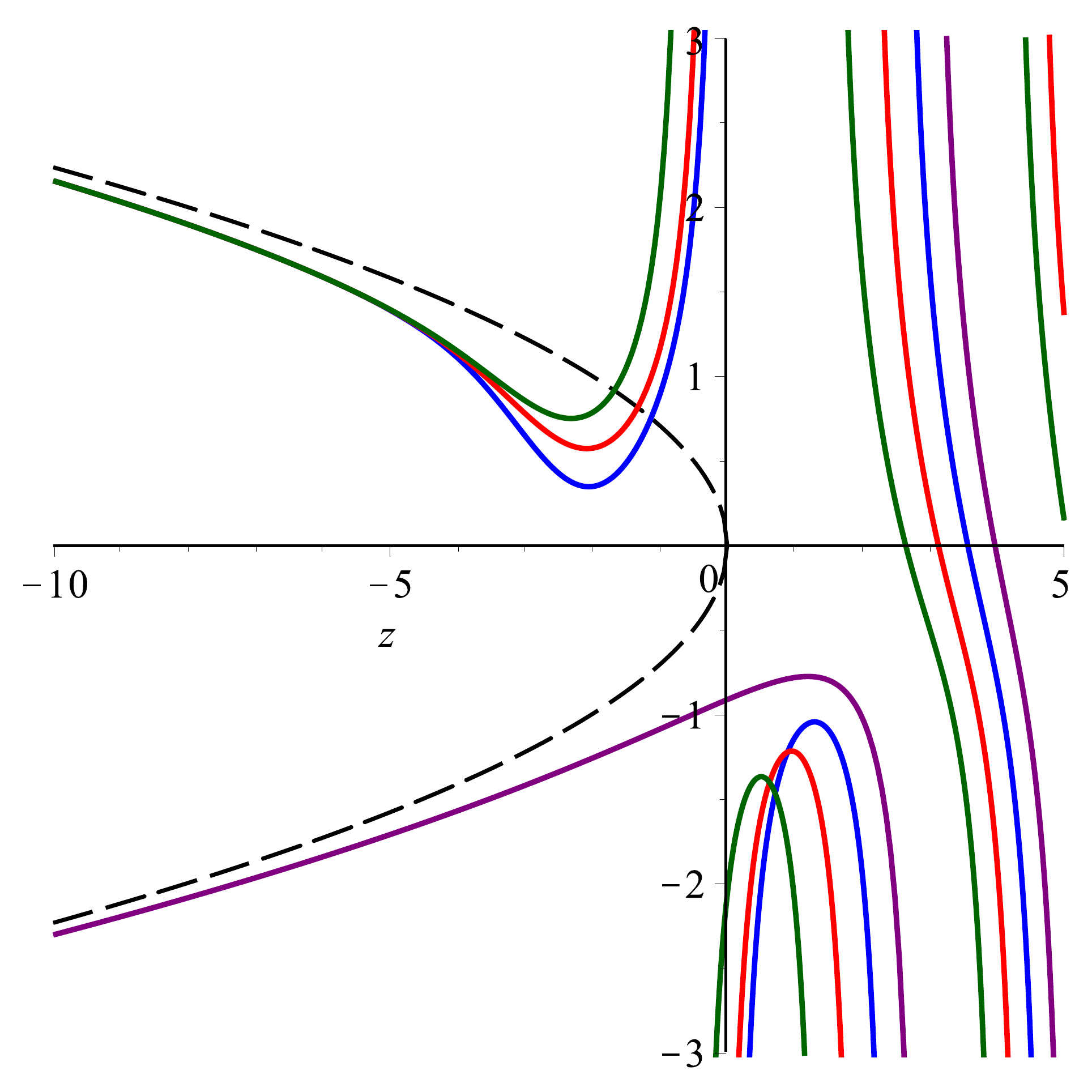}&\fig{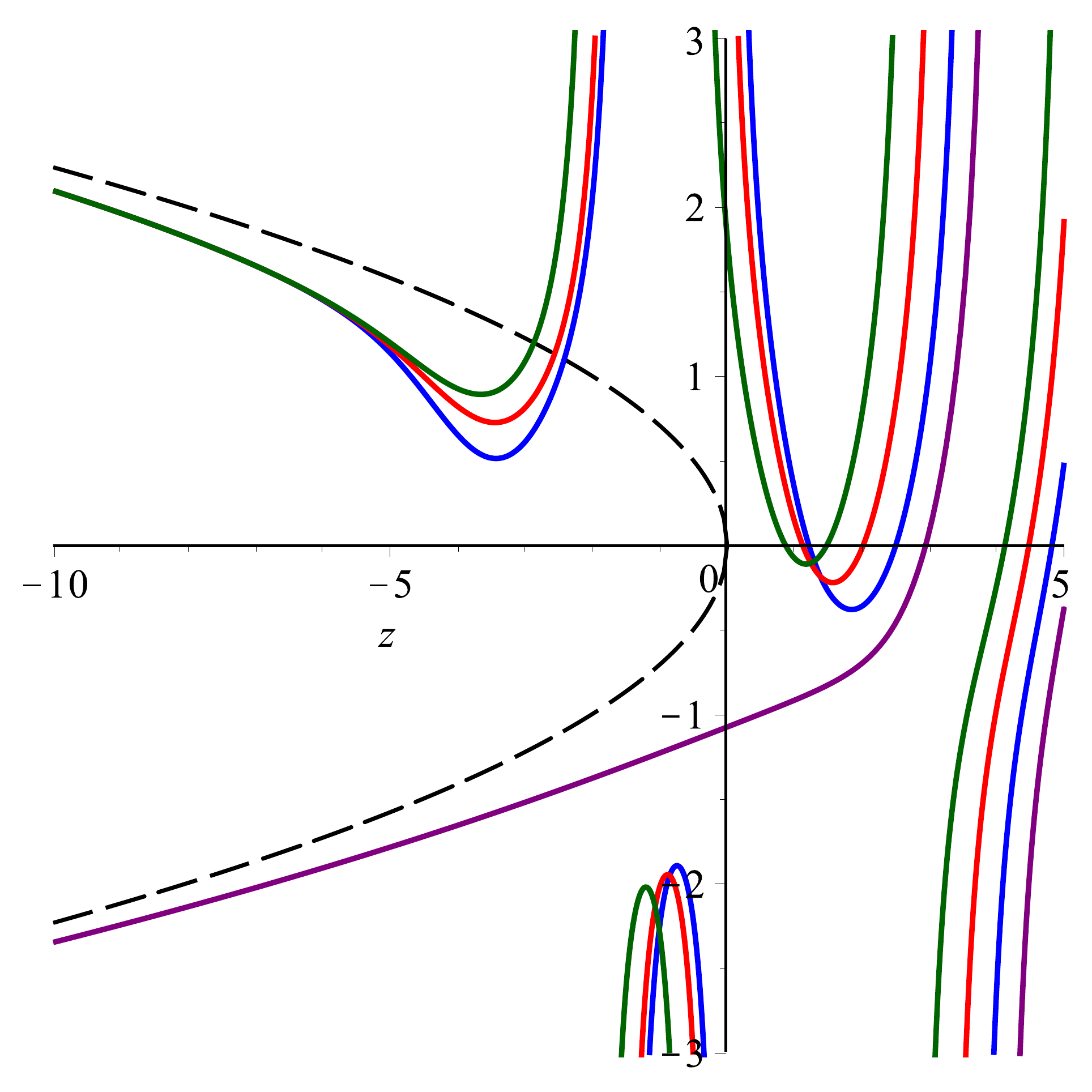}\\
{n=1}
& {n=2}
& {n=3}\\[10pt]
\fig{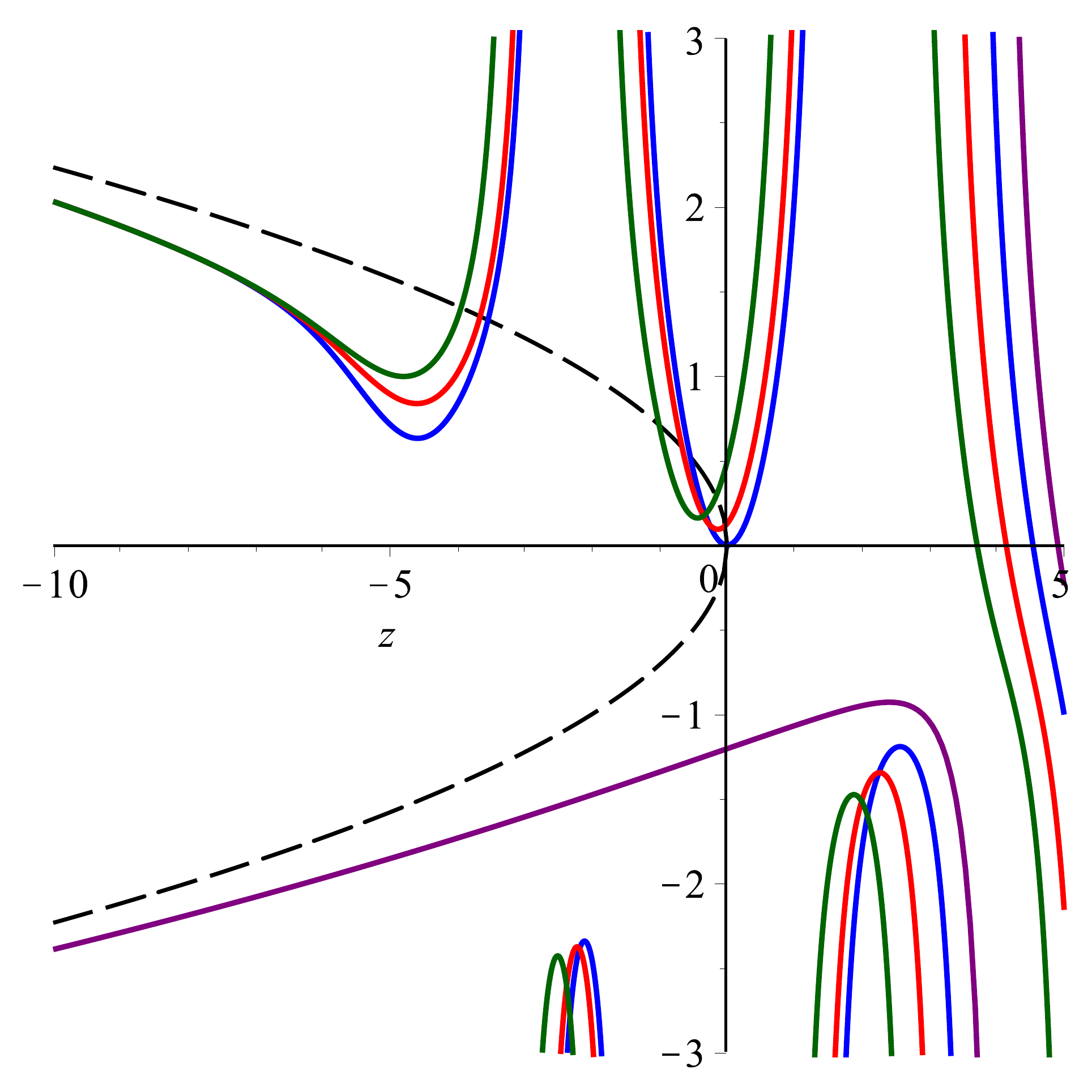} &\fig{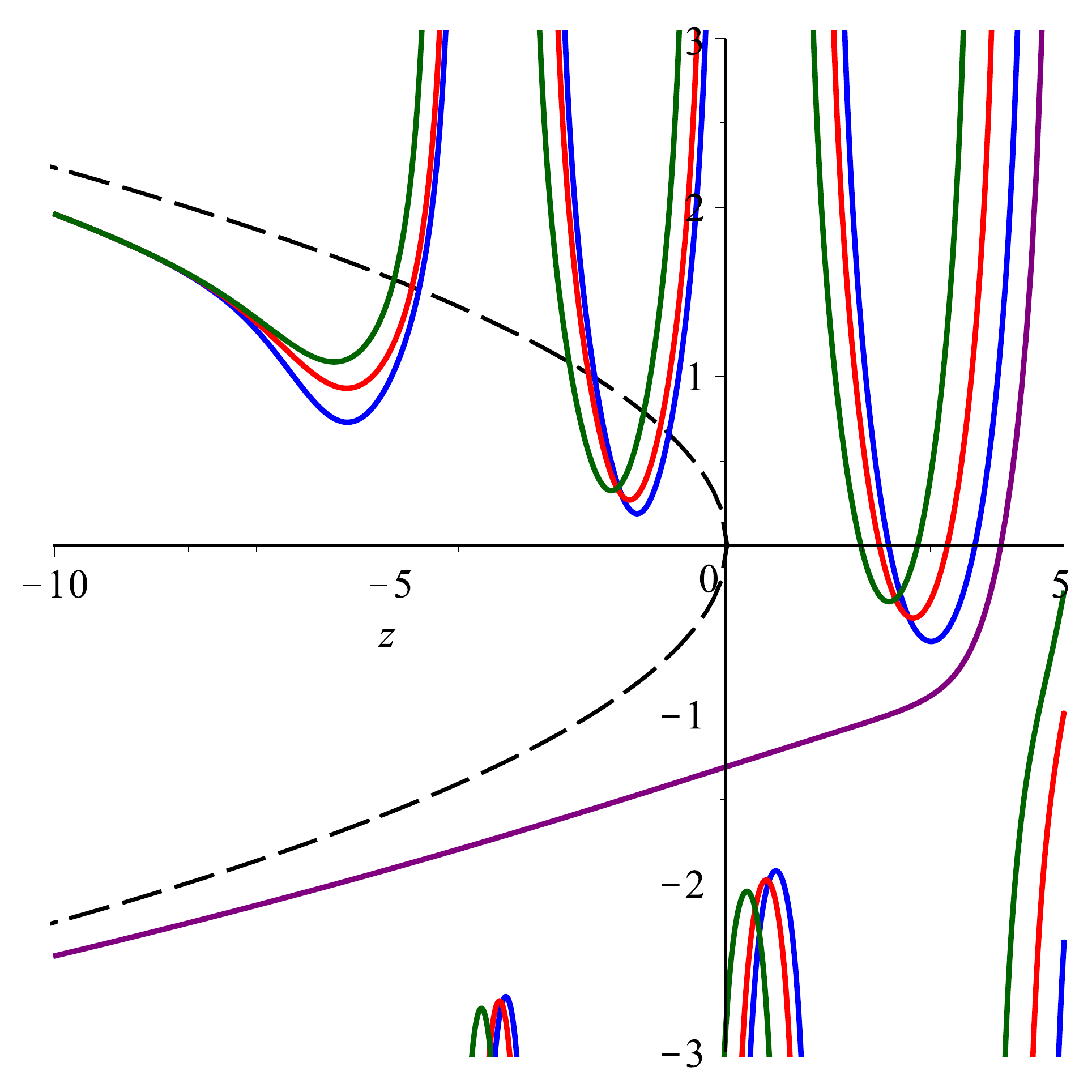}&\fig{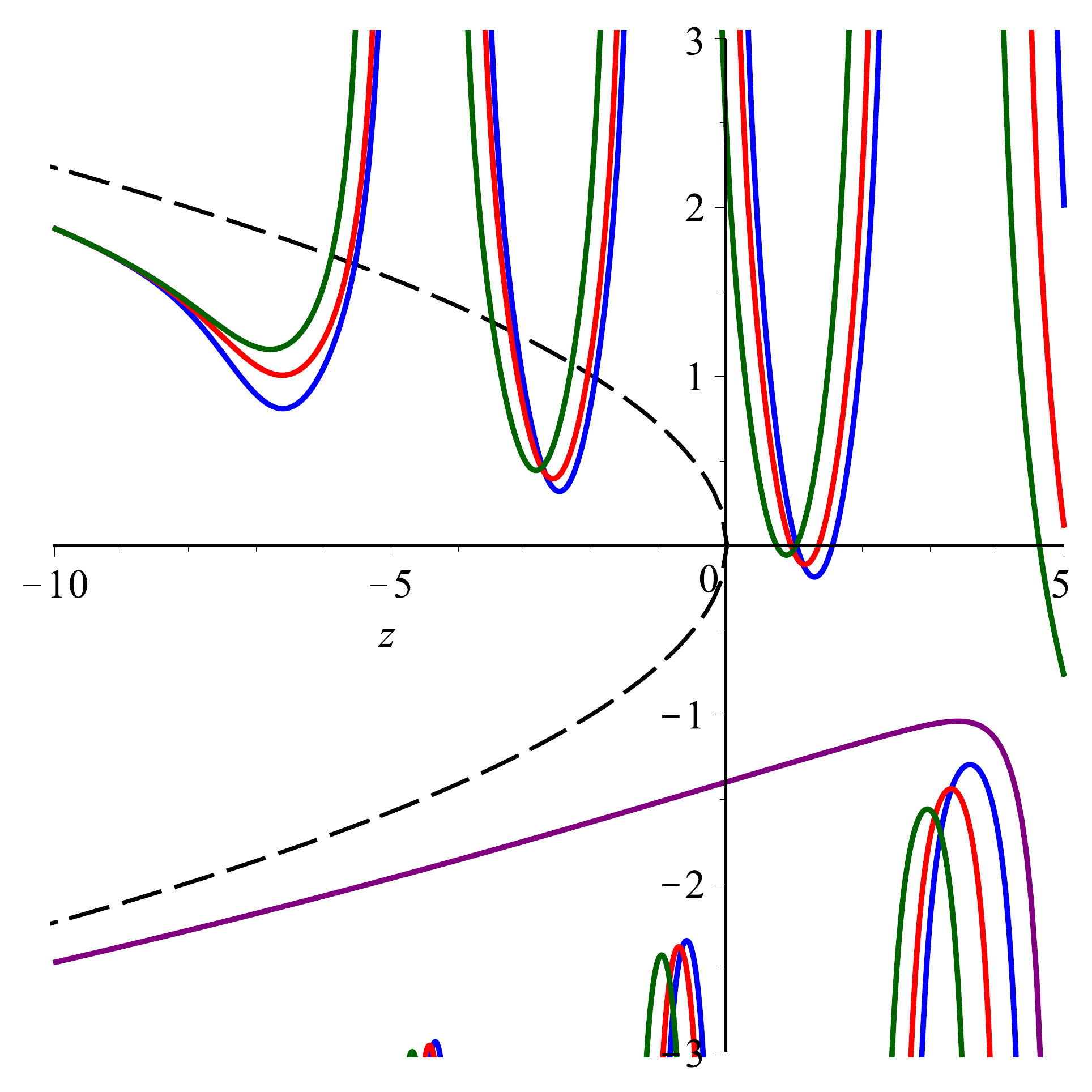}\\
{n=4}
& {n=5}
& {n=6}\end{array}\]
\caption{\label{fig:PIIplots}Plots of $q_n(z;\th)$ (\ref{eq:PT.OP.PII}) for $\th={0}$ [\purple{purple}], $\th={\tfrac16\pi}$ [\blue{blue}], $\th={\tfrac13\pi}$ [\red{red}], $\th={\tfrac12\pi}$ [\green{green}]; the dashed line is the parabola $2q^2+z=0$.}}\end{figure}

{The simplest Airy solutions of \PII\ (\ref{eq:PT.DE.PII}), which arise when $\a=\pm\tfrac12$, are
\begin{equation} q(z) = \mp \deriv{}{z}\ln\varphi(z;\th), \end{equation}
with $\varphi(z;\th)$ given by (\ref{eq:PT.OP.eq23}). Then using the \bts\ (\ref{eq:PT.BT.P22}), a hierarchy of Airy solutions for $\a=n+\tfrac12$, $n\in\Z$, can be generated.}
\begin{table}\[\begin{array}{|c|c|}\hline
n & q_n(z;\th)\\ \hline
1& \phantom{\dfrac{1}{2}} \Phi  \phantom{\dfrac{1}{2}} \\ 
2& -\Phi-\dfrac{1}{2\Phi^2+z}\\[10pt]
3&\dfrac{1}{2\Phi^2+z} - \dfrac{2z\Phi^2-\Phi+z^2}{4\Phi^3+2z\Phi+1}\\[10pt]
4&\dfrac{2z\Phi^2-\Phi+z^2}{4\Phi^3+2z\Phi+1}
-\dfrac{48\Phi^3-8z^2\Phi^2+28z\Phi-4z^3+9}{z(8z\Phi^4-16\Phi^3+8z^2\Phi^2-8z\Phi+2z^3-3)}  -\dfrac{3}{z}\\[10pt] \hline
\end{array}\]
\caption{\label{tab:PII}The Airy solutions $q_n(z;\th)$ of \PII\ (\ref{eq:PT.DE.PII}), see (\ref{eq:PT.OP.PII}), where
$\Phi=- \varphi'(z;\th)/\varphi(z;\th)$, with $\varphi(z;\th)$ given by (\ref{eq:PT.OP.eq23}).}\end{table}
The Airy solutions can also be expressed in terms of determinants, as described in the following theorem.
\begin{theorem}{Let $\tau_n(z;\th)$ be the 
$n\times n$ determinant 
\begin{equation}\label{eq:PT.OP.Ok1}
\tau_n(z;\th) = \left[\deriv[j+k]{}{z}\varphi(z;\th) \right]_{j,k=0}^{n-1},\qquad n\geq1,
\comment{\left|\begin{array}{cccc}
\varphi & \varphi' & \cdots & \varphi^{(n-1)} \\
\varphi' & \varphi'' & \cdots & \varphi^{(n)} \\
\vdots & \vdots & \ddots & \vdots \\
\varphi^{(n-1)} & \varphi^{(n)} & \cdots & \varphi^{(2n-2)} 
\end{array}\right|,}%
\end{equation} 
with $\varphi(z;\th)$ given by (\ref{eq:PT.OP.eq23}) and $\tau_0(z;\th)=1$,
then 
\begin{equation} \label{eq:PT.OP.PII} 
q_n(z;\th)= \deriv{}{z}\ln\frac{\tau_{n-1}(z;\th)}{\tau_n(z;\th)},\qquad n\geq1,
\end{equation} satisfies \PII\ (\ref{eq:PT.DE.PII}) with $\a=n-\tfrac12$. }\end{theorem}

\begin{proof}See 
Flaschka and Newell \cite{refFN}, Okamoto \cite{refOkamotoPIIPIV}; also \cite{refFW01}.\end{proof}

We remark that the determinant 
\[ \Delta_n(t)=\det\left[\deriv[j+k]{}{t}\Ai(t)\right]_{j,k=0}^{n-1},\]
which is equivalent to $\tau_n(z;0)$ given by (\ref{eq:PT.OP.Ok1}), arises in random matrix theory, in connection with the Gaussian Unitary Ensemble (GUE) in the soft-edge scaling limit, see e.g.\ \cite[p.~393]{refFW01}.  

If we set $\Phi(z;\th)\equiv - \varphi'(z;\th)/{\varphi(z;\th)}$, with $\varphi(z)$ given by (\ref{eq:PT.OP.eq23}), then the first few solutions in the Airy function solution hierarchy for \PII\ (\ref{eq:PT.DE.PII}) are given in Table \ref{tab:PII}. We note that $\Phi(z;\th)$ satisfies the Riccati equation (\ref{eq:PT.OP.eq21}) with $\ep=1$.

Plots of the solutions $q_n(z;\th)$ (\ref{eq:PT.OP.PII}) 
for various $\th$ are given in Figure \ref{fig:PIIplots}. These plots show that the asymptotic behaviour as $z\to-\infty$ of the solutions is completely different in the case when $\th=0$ compared to the case when $\th\not=0$, see Theorem \ref{thm33}.
\comment{\begin{lemma}\label{lem42} If $q_n(z;\th)$ is given by (\ref{eq:PT.OP.PII}), then
\begin{equation} q_n(z;\th)=\begin{cases} -\tfrac12\sqrt {2}\,(-z)^{1/2}+ \O\big(z^{-1}\big), &\quad{\rm if}\quad\th=0,\\
\tfrac12\sqrt {2}\,(-z)^{1/2}+ \O\big(z^{-1}\big), &\quad{\rm if}\quad\th\not=0.\end{cases}\end{equation}
\end{lemma}
\begin{proof}See\end{proof}}
Fornberg and Weideman \cite[Figure 3]{refFW} plot the locations of the poles for the solution $q_3(z;\th)$ of \PII\ (\ref{eq:PT.DE.PII}), i.e.\ when $\a=\tfrac52$, for various choices of $\th$, which show that the pole structure of the solutions is significantly different in the case when $\th=0$ compared to the case when $\th\not=0$. 
The solution $q_3(z;0)$, which depends only on $\Ai(t)$, is a \textit{tronqu\'{e}e} solution, i.e.\ it has no poles in a sector of the complex plane; such solutions play an important role in the theory of \p\ equations. 
Numerical calculations suggest that the special solutions $q_n(z;0)$ are \textit{tronqu\'{e}e} solutions, see Figure \ref{fig:PIIpoles}. 

\begin{figure}
\[\begin{array}{c@{\quad}c@{\quad}c}
\fig{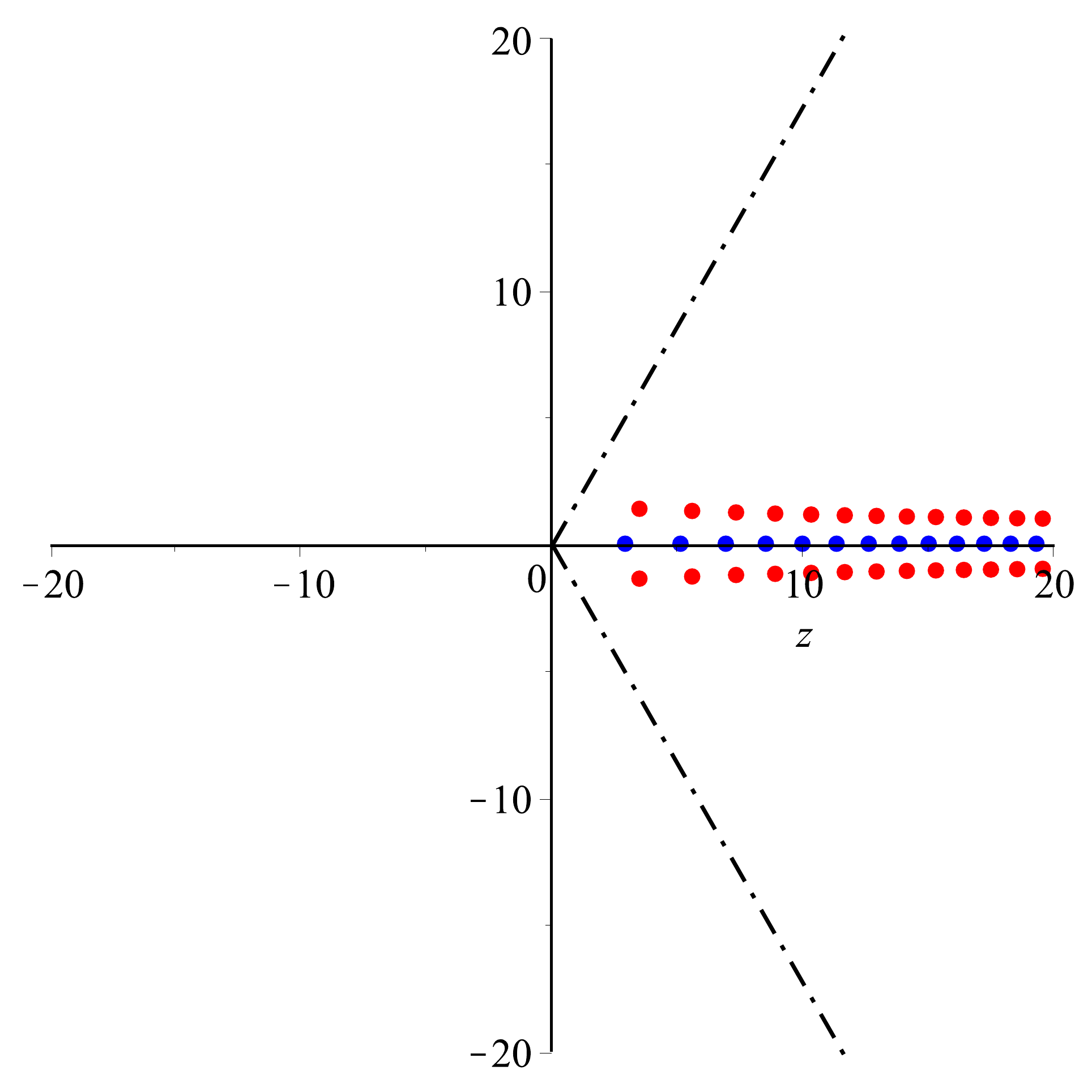}&\fig{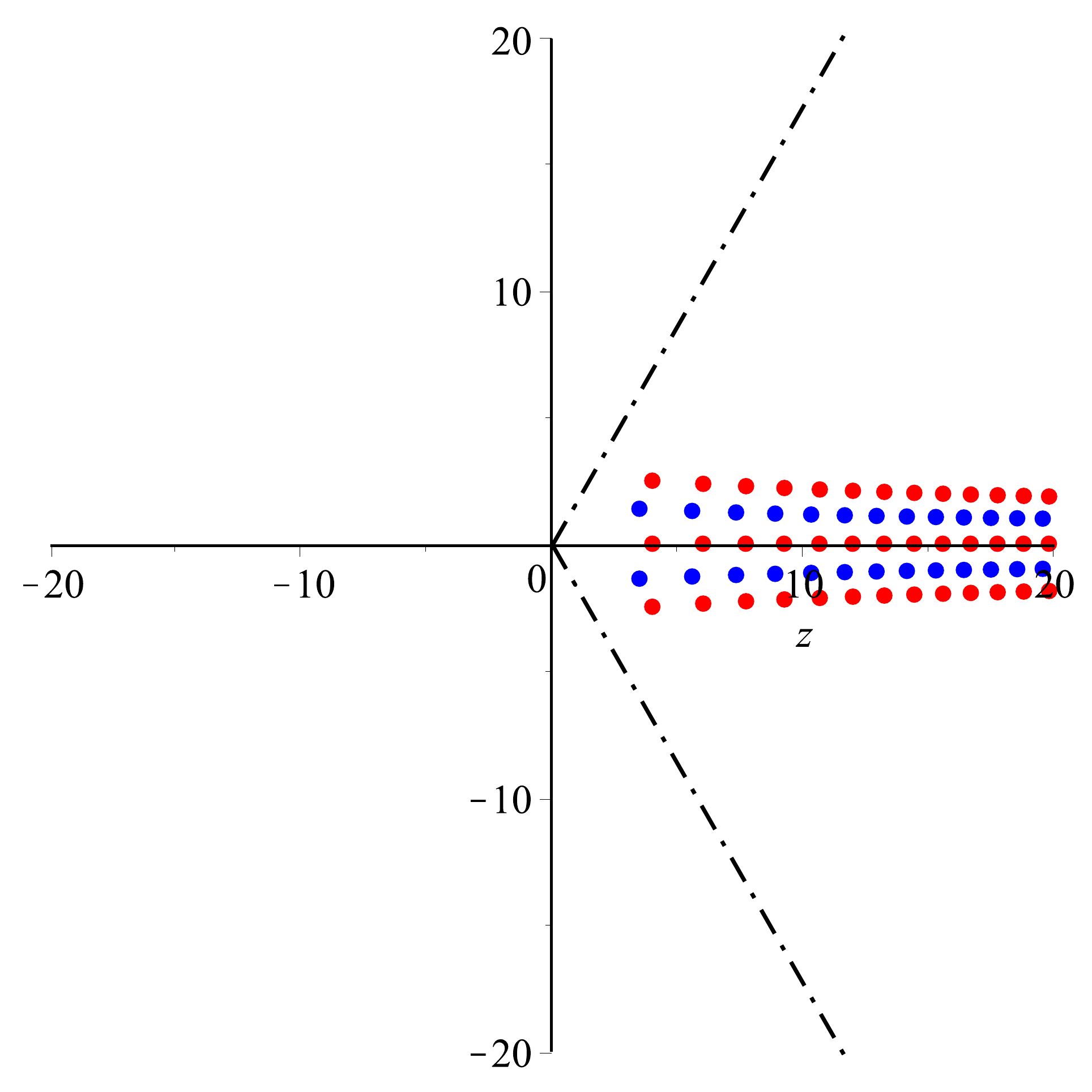}&\fig{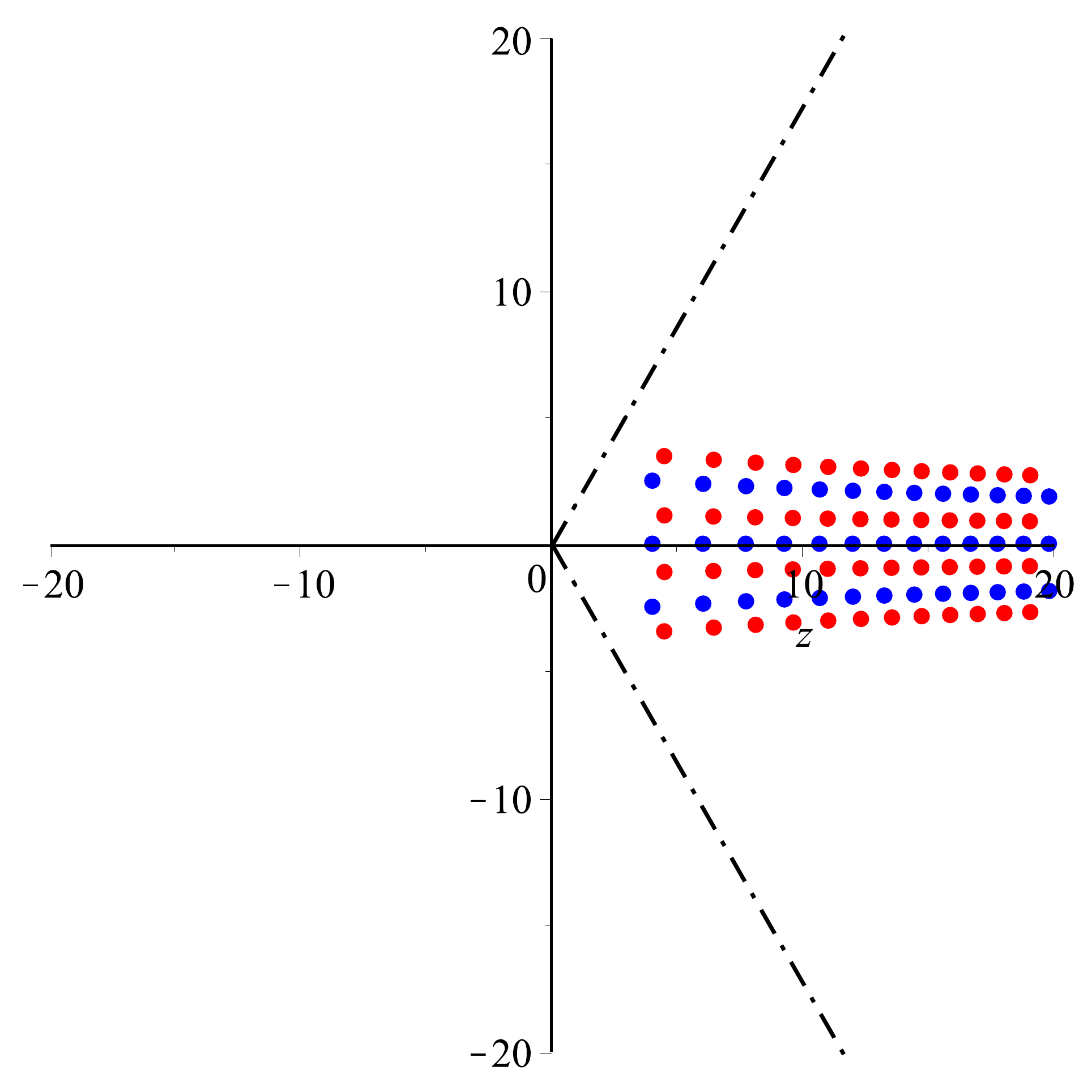}\\
n=2 & n=3 & n=4
\end{array}\]
\caption{\label{fig:PIIpoles}Plots of the poles of $q_n(z;0)$ and $\sigma_n(z;0)$ for $n=2,3,4$; the blue and red circles represent poles with residues $+1$ and $-1$, respectively.}\end{figure}

\comment{\begin{figure}{\[ \begin{array}{c@{\quad}c@{\quad}c}
\fig{PIIq1asymp} &\fig{PIIq1asymp}&\fig{PIIq1asymp}\\
{n=1}
& {n=2}
& {n=3}\end{array}\]
\caption{\label{fig:PIIplots}Plots of $q_n(z;\th)$ (\ref{eq:PT.OP.PII}) for $\th={0}$ [\purple{purple}], $\th={\tfrac1{1000}\pi}$ [\green{green}] $\th={\tfrac1{100}\pi}$ [\blue{blue}], $\th={\tfrac1{10}\pi}$ [\red{red}], $\th={\tfrac12\pi}$ [black]; the dashed line is the parabola $2q^2+z=0$.}}\end{figure}}%

When $\th=0$ the solution $q_n(z;0)$ involves only on the Airy function $\Ai(t)$, with $t=-2^{-1/3}z$, whereas the solution $q_n(z;\th)$, with $\th\not=0$, involves the Airy function $\Bi(t)$. The known asymptotics of $\Ai(t)$ and $\Bi(t)$ as $t\to\infty$ are
\begin{subequations}\label{asymp:AiBi}\begin{align}
\Ai(t)&=\tfrac12 {\pi}^{-1/2}\,t^{-1/4}\exp\big(-\tfrac23t^{3/2}\big)\left\{1+ \O\big(t^{-3/2}\big)\right\},\\
\Bi(t)&={\pi}^{-1/2}\,t^{-1/4}\exp\big(\tfrac23t^{3/2}\big)\left\{1+ \O\big(t^{-3/2}\big)\right\},
\end{align}\end{subequations} cf.~\cite[\S9.7(ii)]{refDLMF}. Consequently the asymptotic behaviour of $q_n(z;\th)$ as $z\to-\infty$ critically depends on whether it involves $\Bi(t)$.

The asymptotic behaviour of the Airy solutions $q_n(z;\th)$ as $z\to-\infty$ is given in the following theorem.
\begin{theorem}{\label{thm33}Let $q_n(z;\th)$ be defined by (\ref{eq:PT.OP.PII}), then as $z\to-\infty$, 
\begin{align*}
q_n(z;\th)&=\begin{cases}
\ds-\frac{(-z)^{1/2}}{\sqrt {2}}+\frac {2n-1}{4z}
+{\frac {12n^2-12n+5}{16\sqrt {2}\,(-z)^{5/2}}}+\O\left(z^{-4}\right),&\mbox{if}\enskip \th=0,\\
\ds\frac{(-z)^{1/2}}{\sqrt {2}}+\frac {2n-1}{4z}-{\frac {12n^2-12n+5}{16\sqrt {2}\,(-z)^{5/2}}}
+\O\left(z^{-4}\right),&\mbox{if}\enskip \th\not=0.\end{cases}
\end{align*}
}\end{theorem} 

\begin{proof}
These are proved using the asymptotics (\ref{asymp:AiBi}) of $\Ai(t)$ and $\Bi(t)$ as $t\to\infty$;
for details, see Clarkson, Loureiro, and Van Assche \cite{refCLvA}.\end{proof}

The plots in Figure \ref{fig:PIIplots} suggest the following conjecture (see also \cite{refCLvA}).
\begin{conjecture}{\label{thm34}If $q_n(z;0)$ is defined by (\ref{eq:PT.OP.PII}), then for $z<0$ and $n\geq1$, $q_n(z;0)$ is a monotonically decreasing function and \[q_{n+1}(z;0)<q_n(z;0).\]}\end{conjecture}

\subsection{\p\ XXXIV equation and Jimbo-Miwa-Okamoto $\sigma$ equation}
Due to the one-to-one correspondence between solutions of \PII\ (\ref{eq:PT.DE.PII}) and those of \Ptf\
(\ref{eq:PT.DE.P34}) and \SII\ (\ref{eq:PT.DE.SII}), as discussed in \S\ref{ssec:Ham}, we have the following result.

\begin{theorem}{Equations (\ref{eq:PT.DE.P34}) and (\ref{eq:PT.DE.SII}) possess
one-parameter family of solutions expressible in terms of Airy functions given by (\ref{eq:PT.OP.eq23}) if and only if
$\a=n-\tfrac12$, with $n\in\Integer$.}\end{theorem}

\begin{figure}{\[ \begin{array}{c@{\quad}c@{\quad}c}
\fig{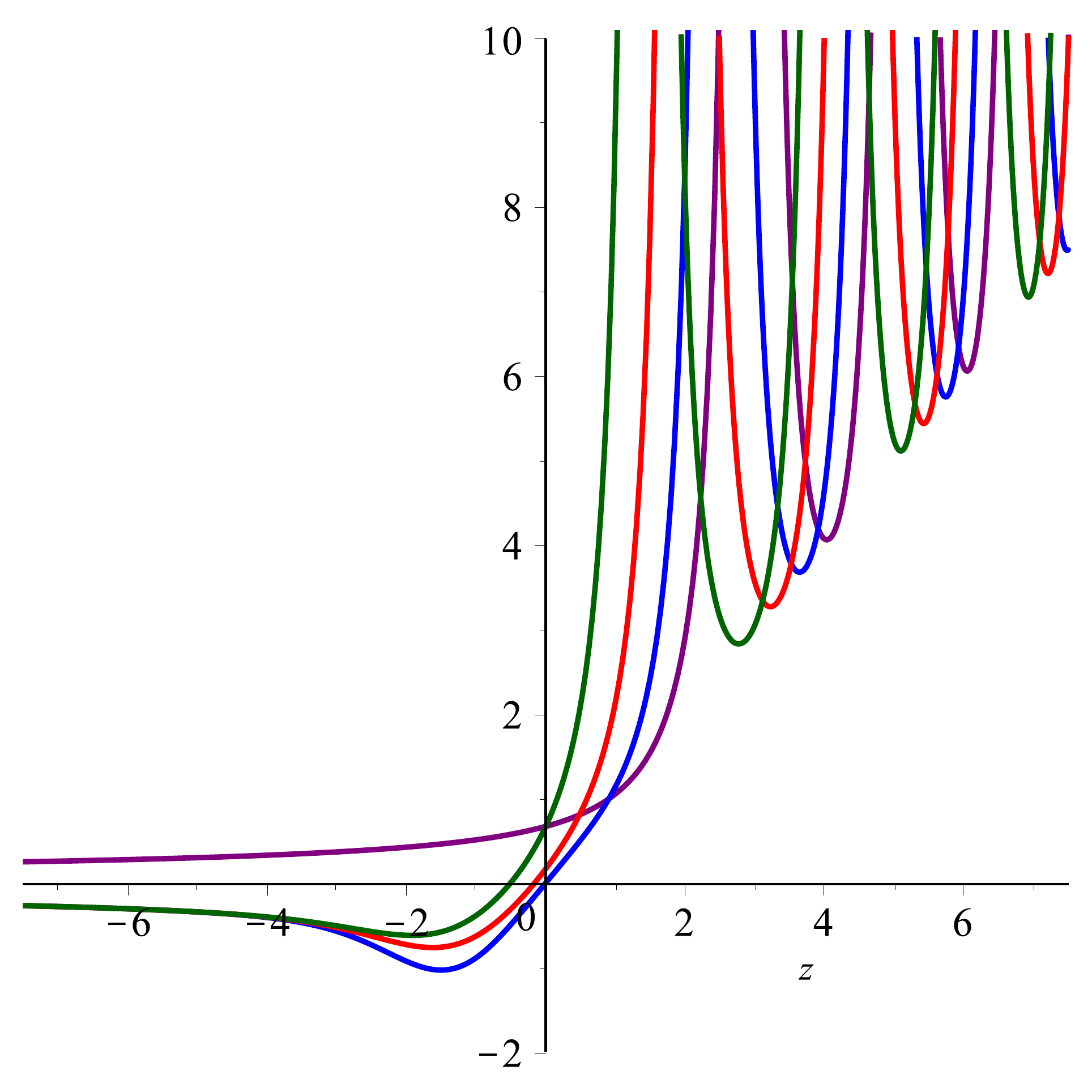} &\fig{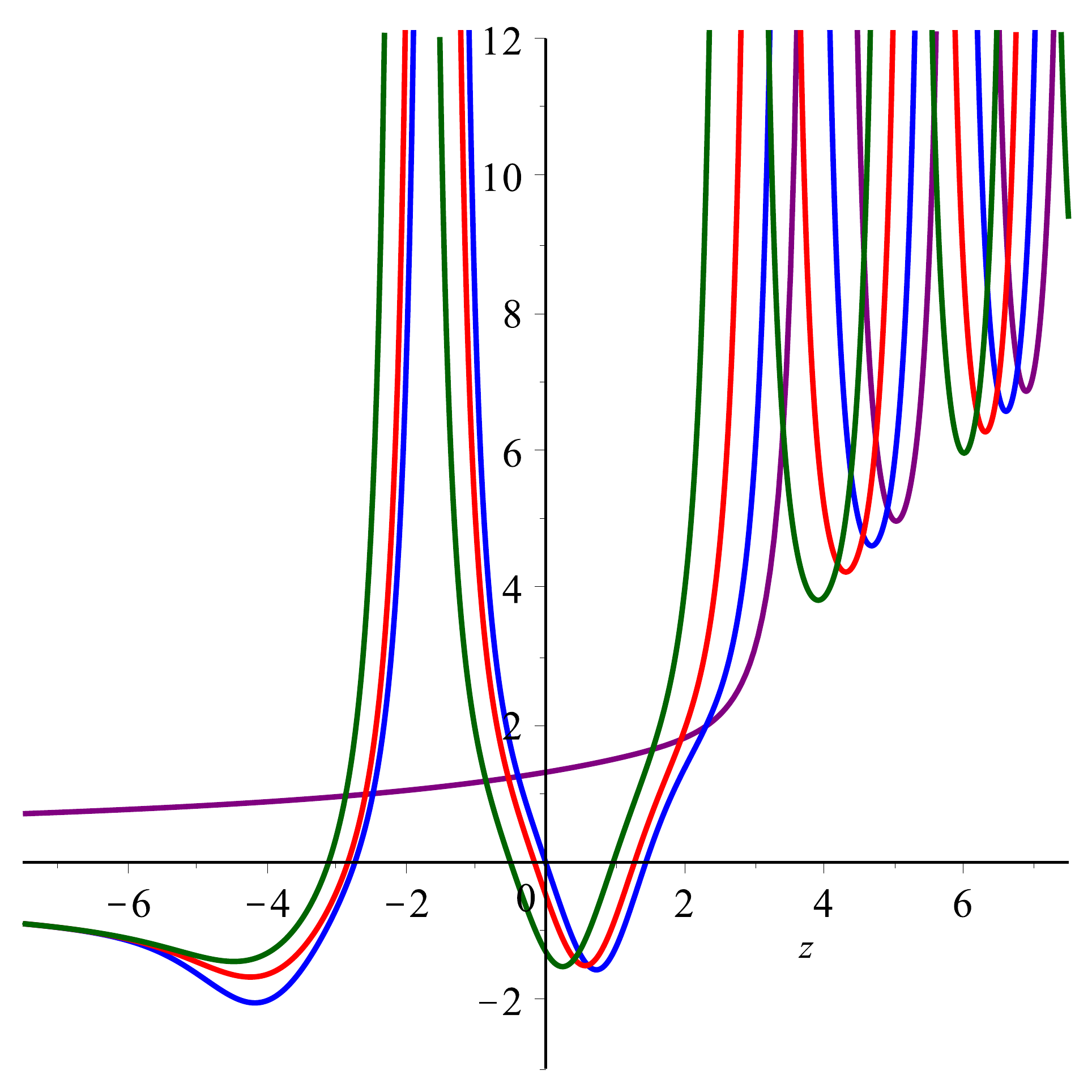} &\fig{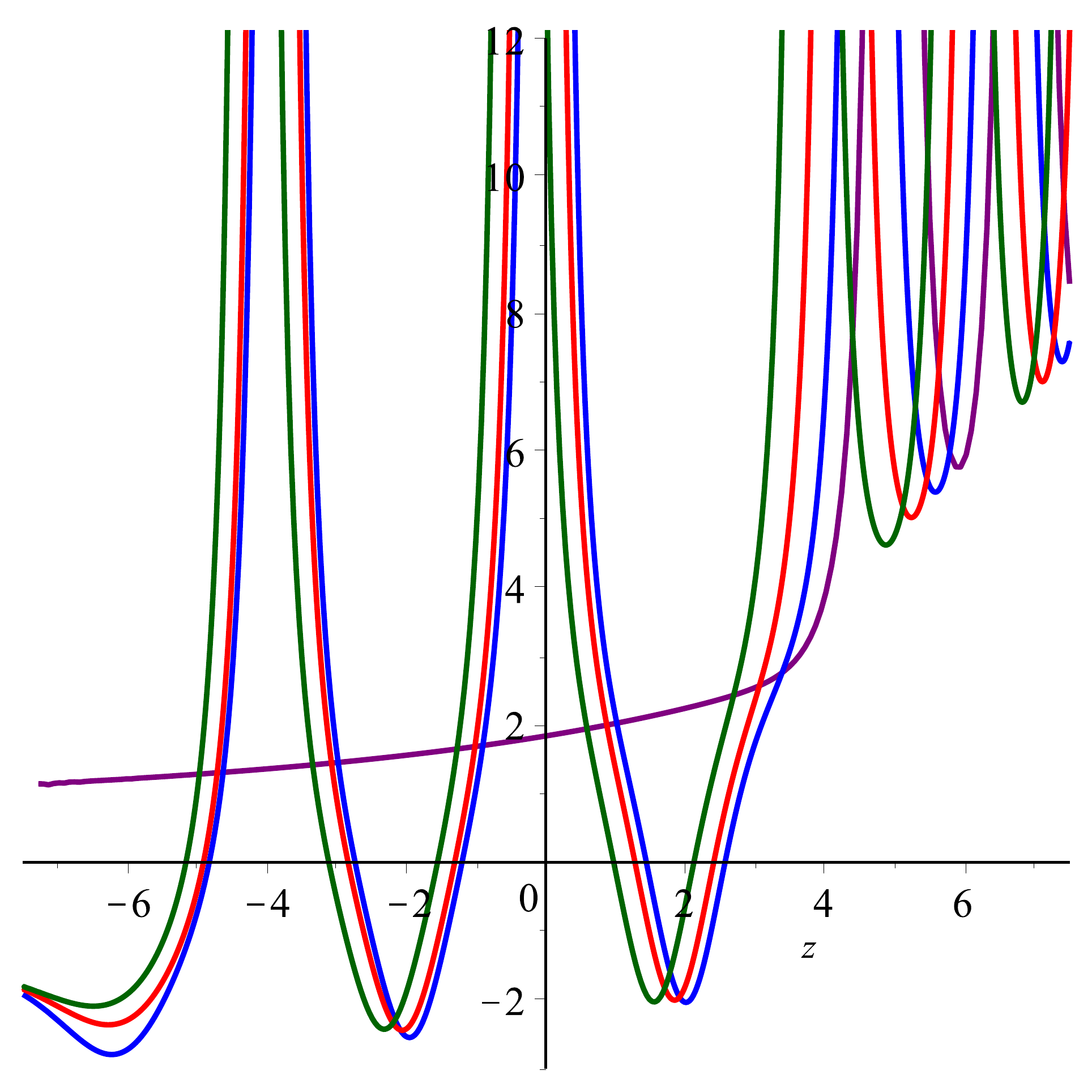}\\
{n=1}
& {n=3}
& {n=5}\\[10pt]
\fig{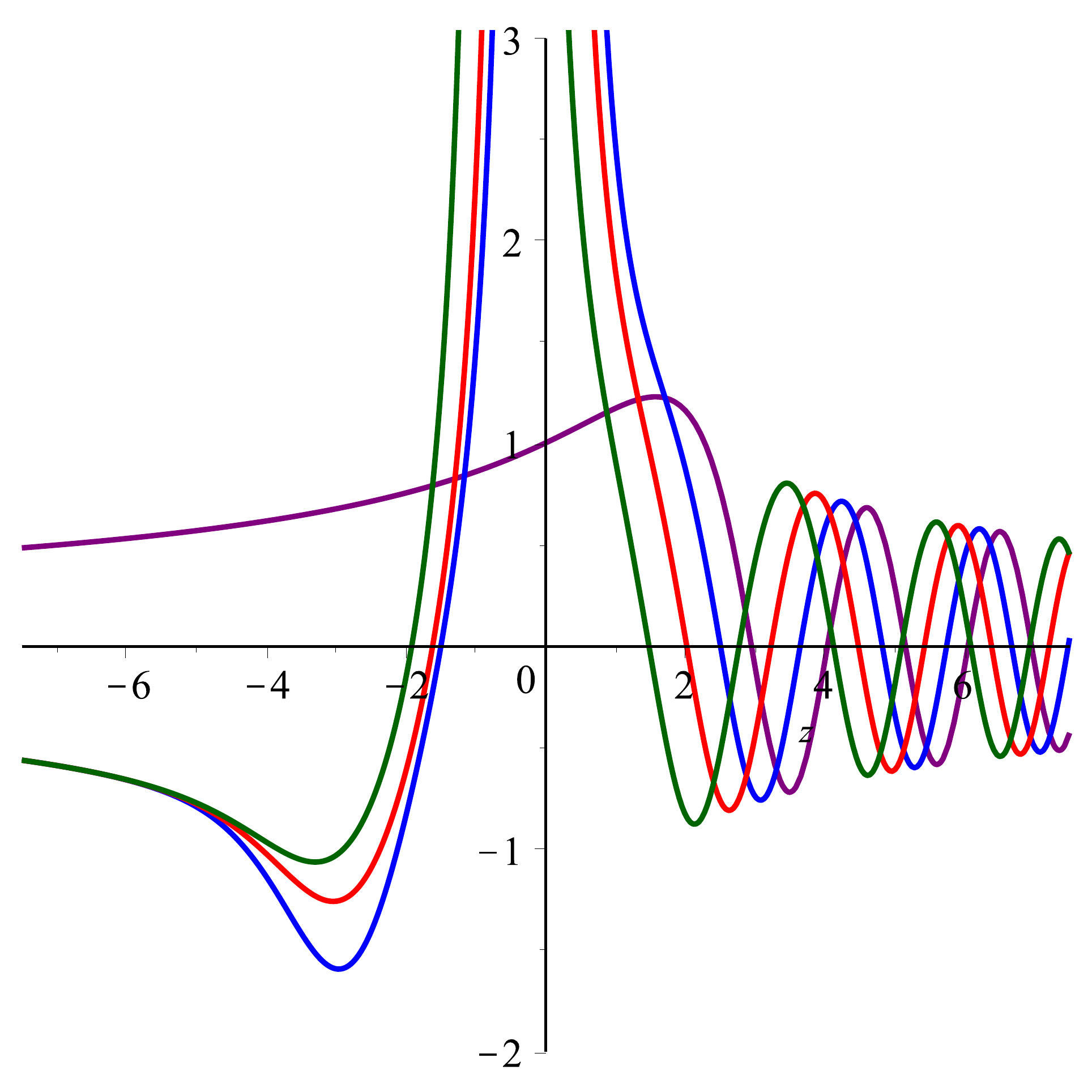} &\fig{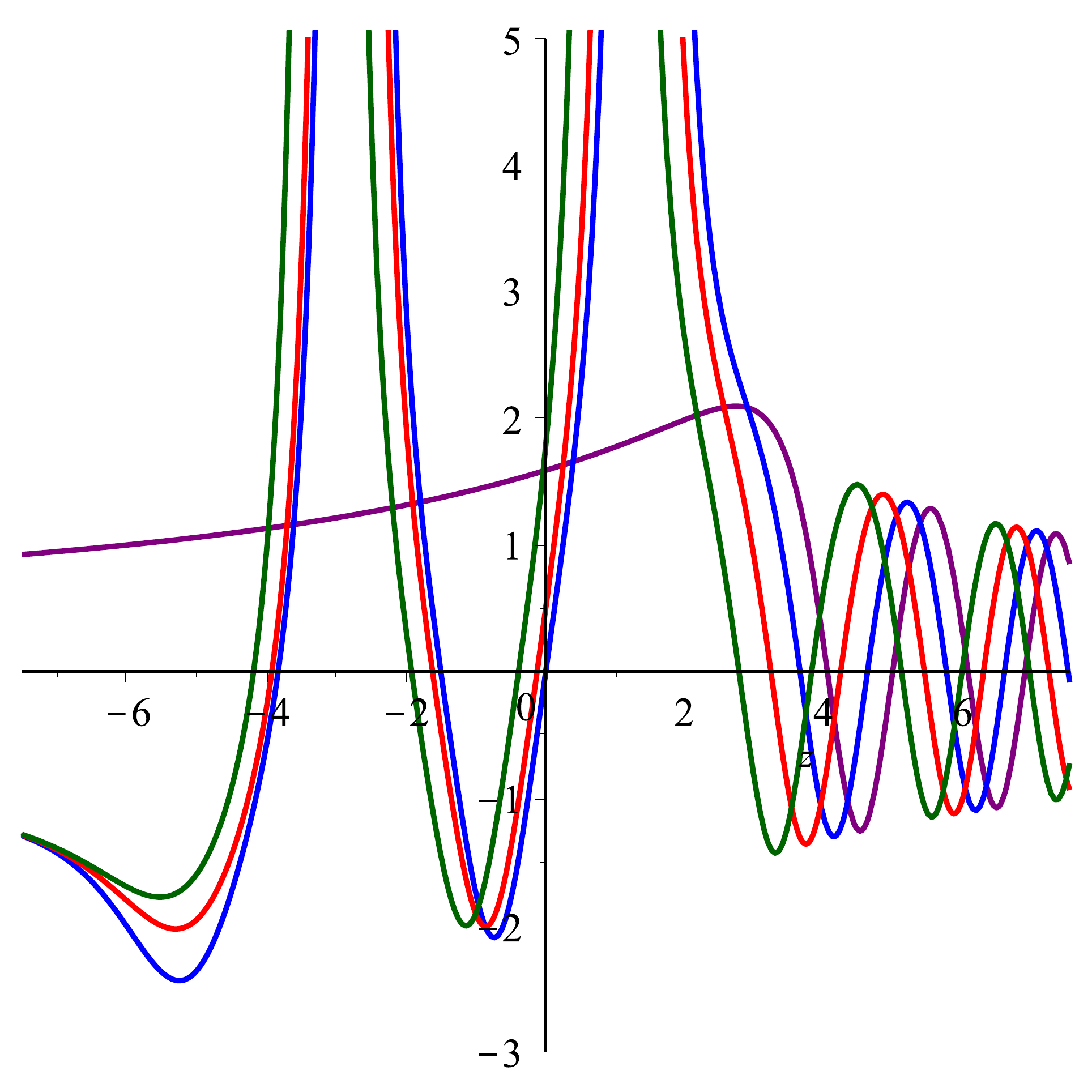}&\fig{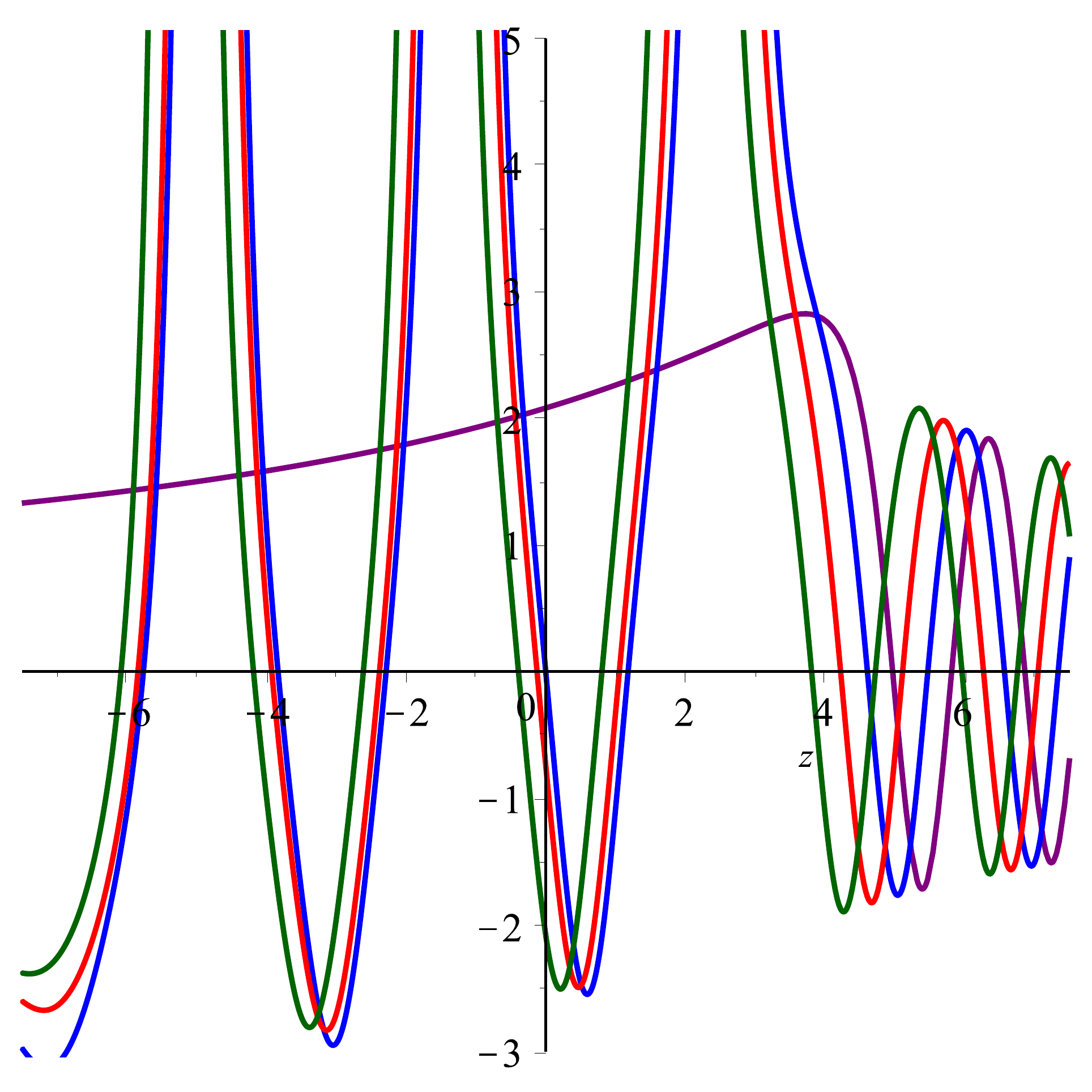}\\
{n=2}
& {n=4}
& {n=6}\end{array}\]
\caption{\label{fig:P34plots1}Plots of $p_n(z;\th)/n$, with $p_n(z;\th)$ given by (\ref{eq:PT.OP.P34}), for $\th={0}$ [\purple{purple}], $\th={\tfrac16\pi}$ [\blue{blue}], $\th={\tfrac13\pi}$ [\red{red}], $\th={\tfrac12\pi}$ [\green{green}].}}\end{figure}
\begin{table}\[\begin{array}{|c|c|}\hline
n & p_n(z;\th) \\ \hline
1& \phantom{\dfrac{1}{2}}  2\Phi^2+z \phantom{\dfrac{1}{2}} \\
2&\dfrac{4\Phi}{2\Phi^2+z}+\dfrac{2}{(2\Phi^2+z)^2}\\[10pt]
3&z-\dfrac{8\Phi^2-2z^2\Phi+6z}{4\Phi^3+2z\Phi+1}+\dfrac{2(2z^2+1)\Phi^2-2z^2\Phi+2z(z^3-1)}{(4\Phi^3+2z\Phi+1)^2}\\[10pt]\hline
\end{array}\]
\caption{\label{tab:P34}The Airy solutions $p_n(z;\th)$ of \Ptf\ (\ref{eq:PT.DE.P34}) given by (\ref{eq:PT.OP.P34}).}\end{table}

\begin{figure}{\[ \begin{array}{c@{\quad}c@{\quad}c}
\fig{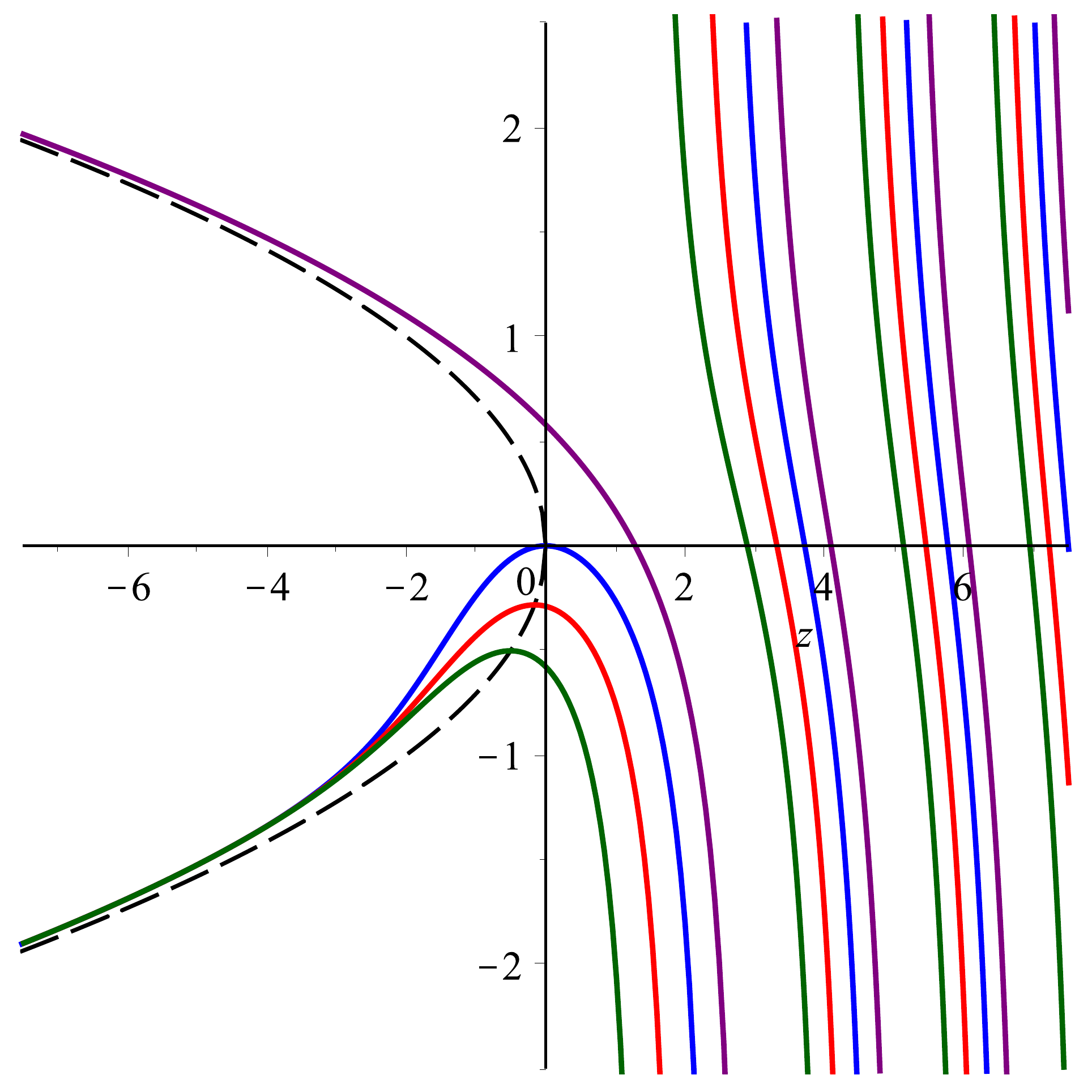} &\fig{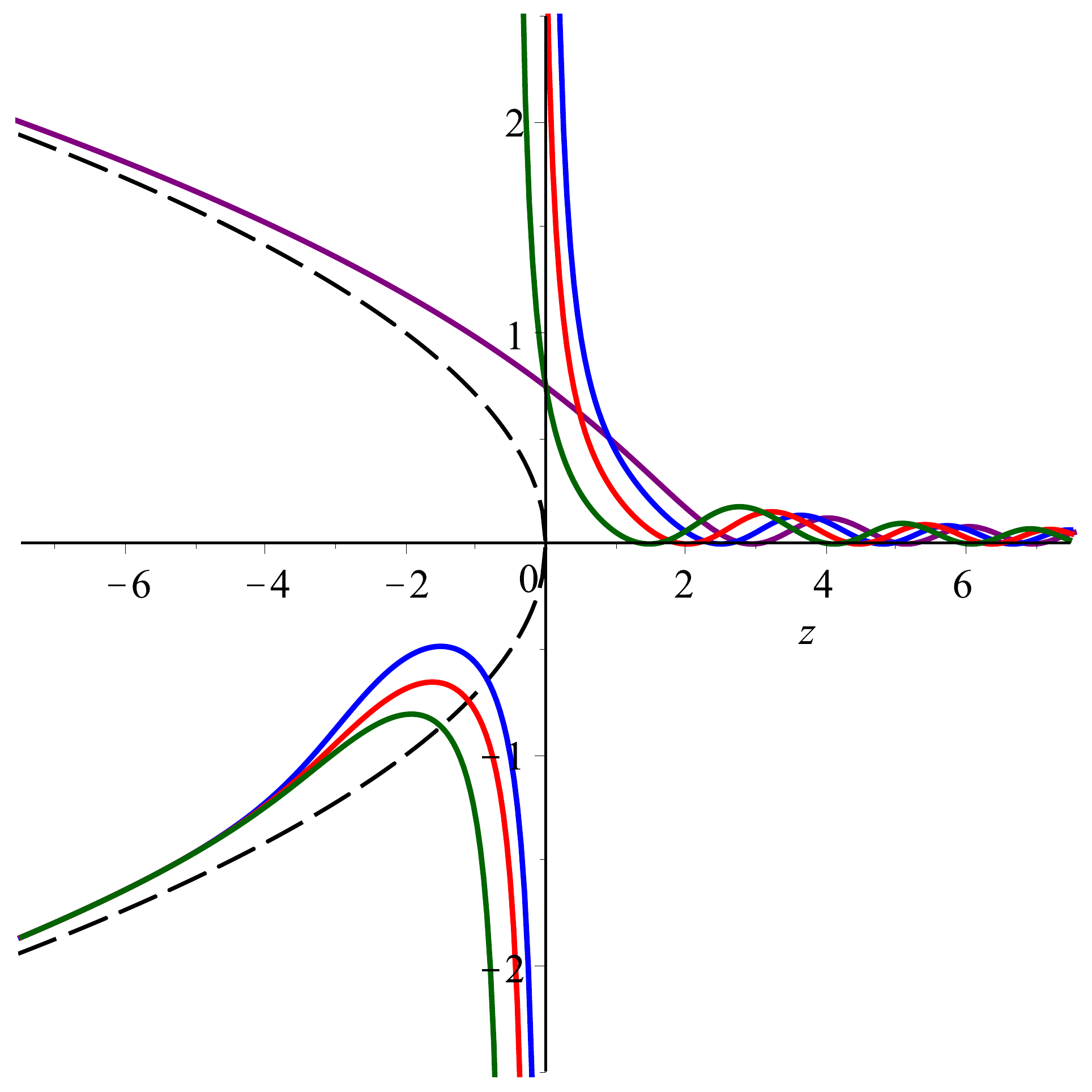} & \fig{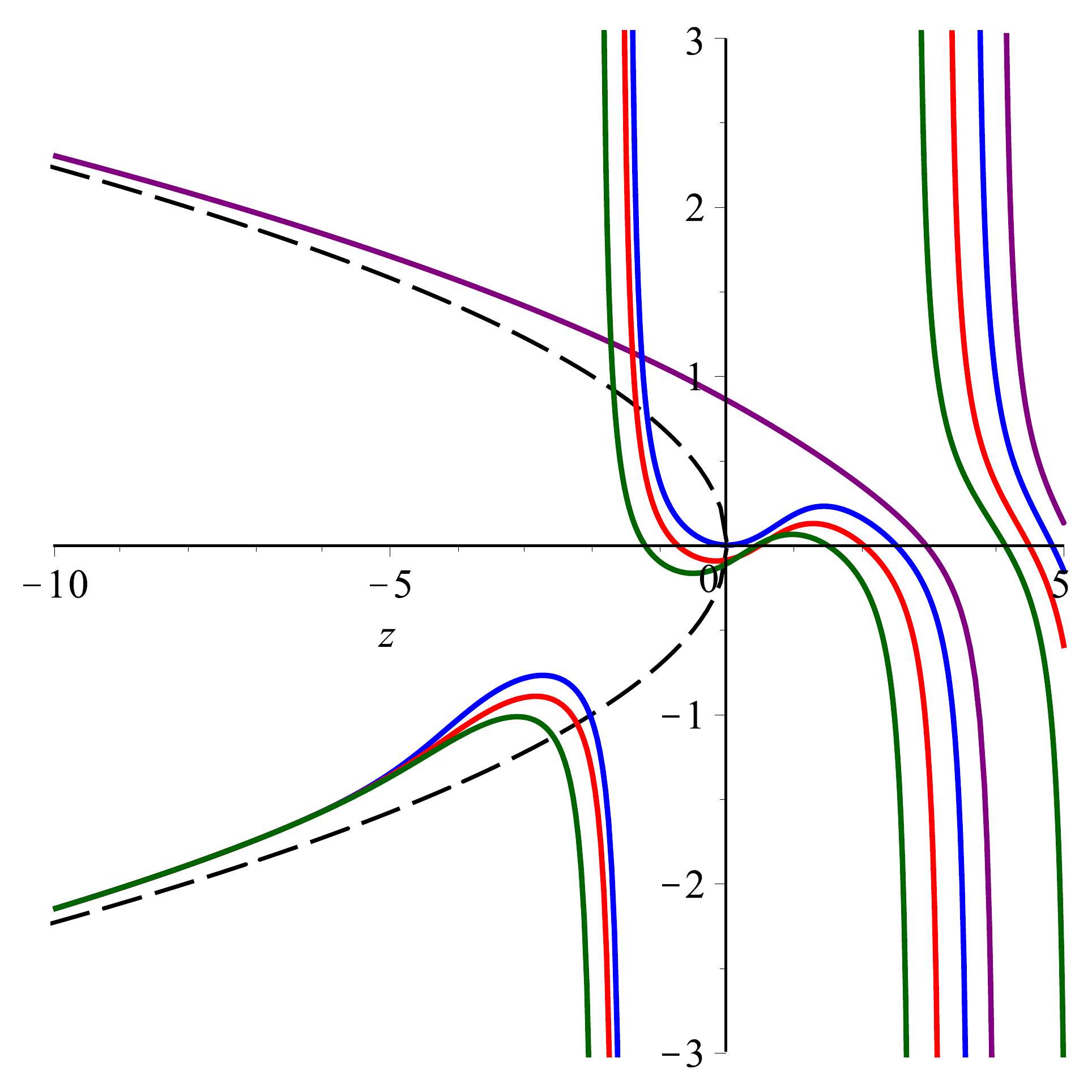}\\
{n=1}
& {n=2}
&{n=3}\\[10pt]
\fig{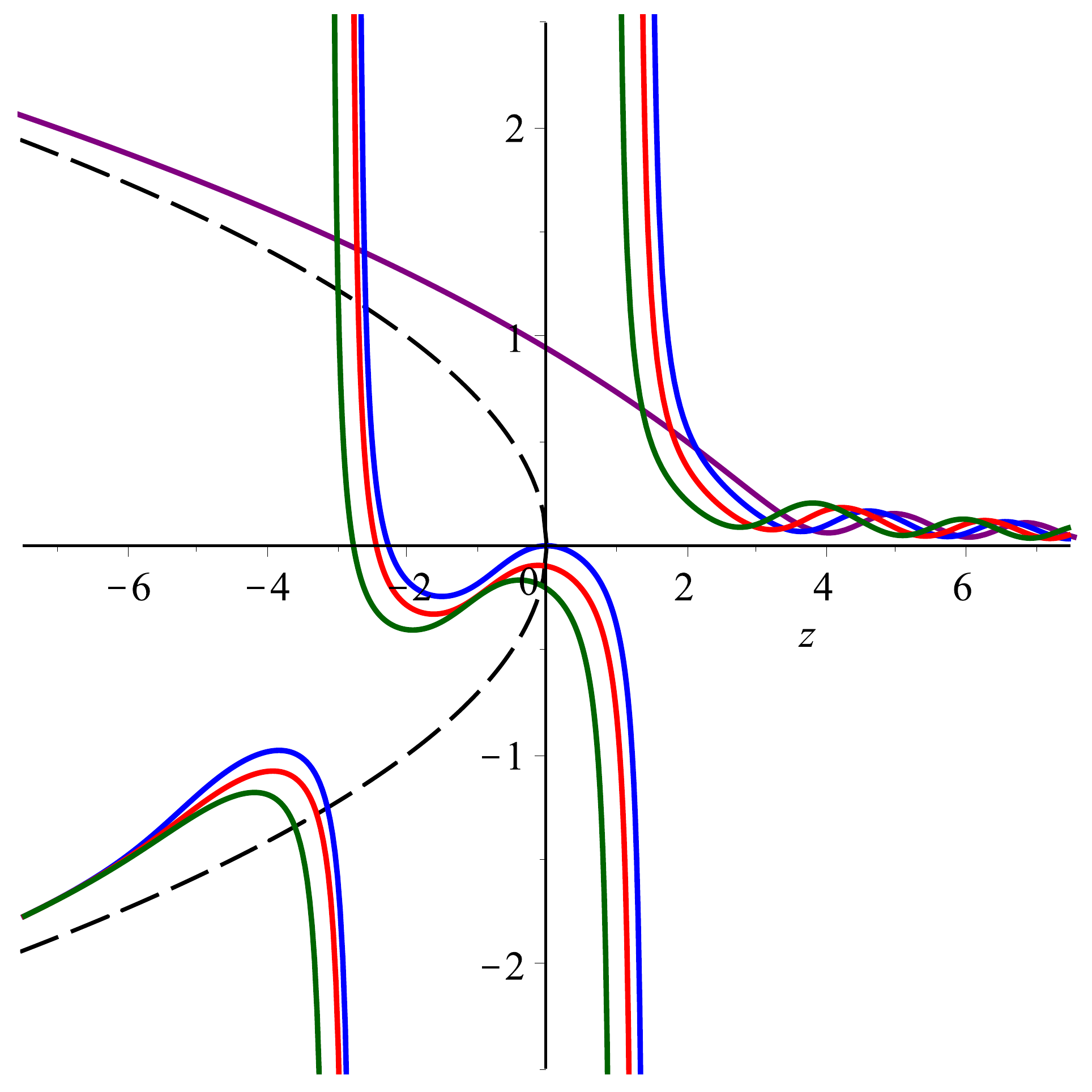} &\fig{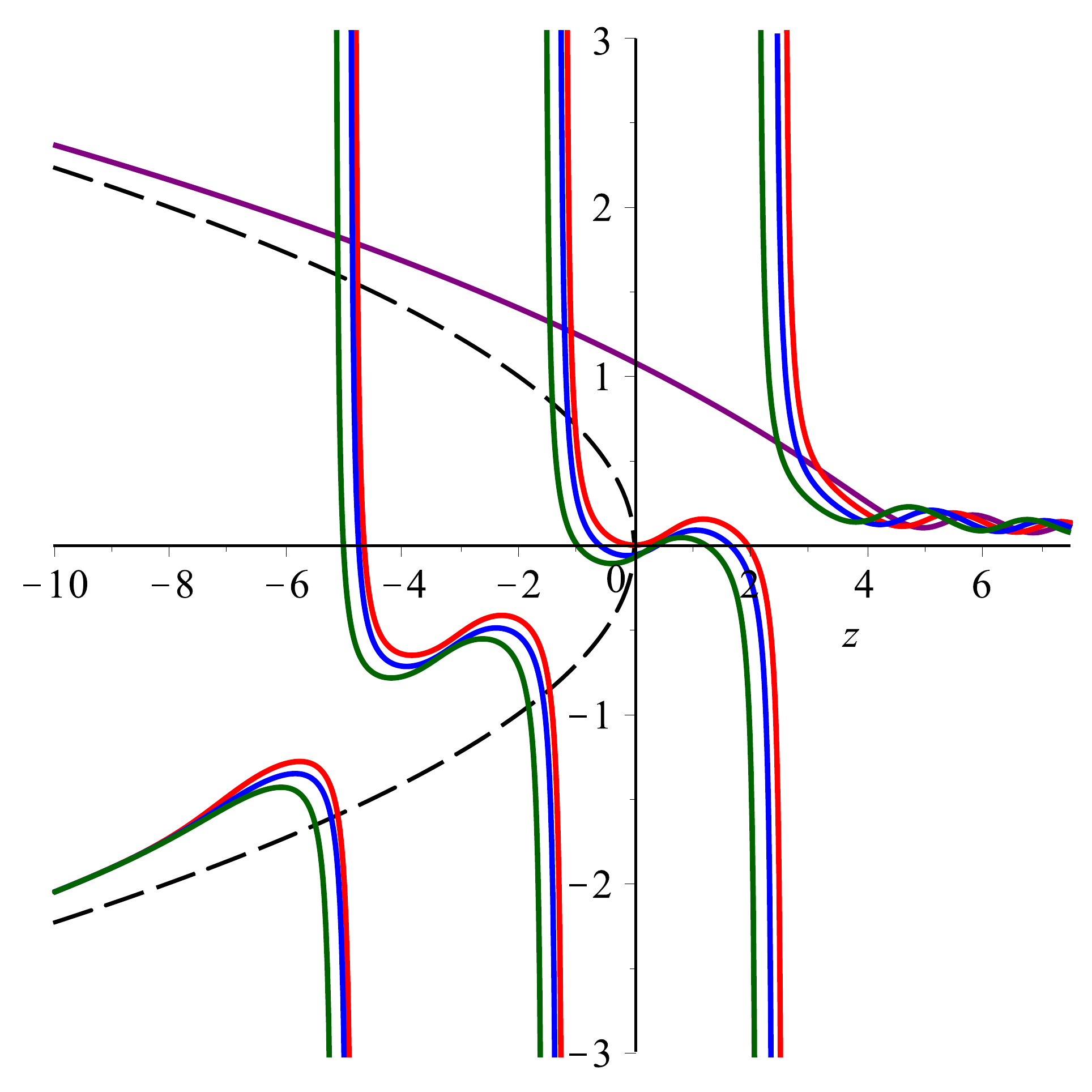}&\fig{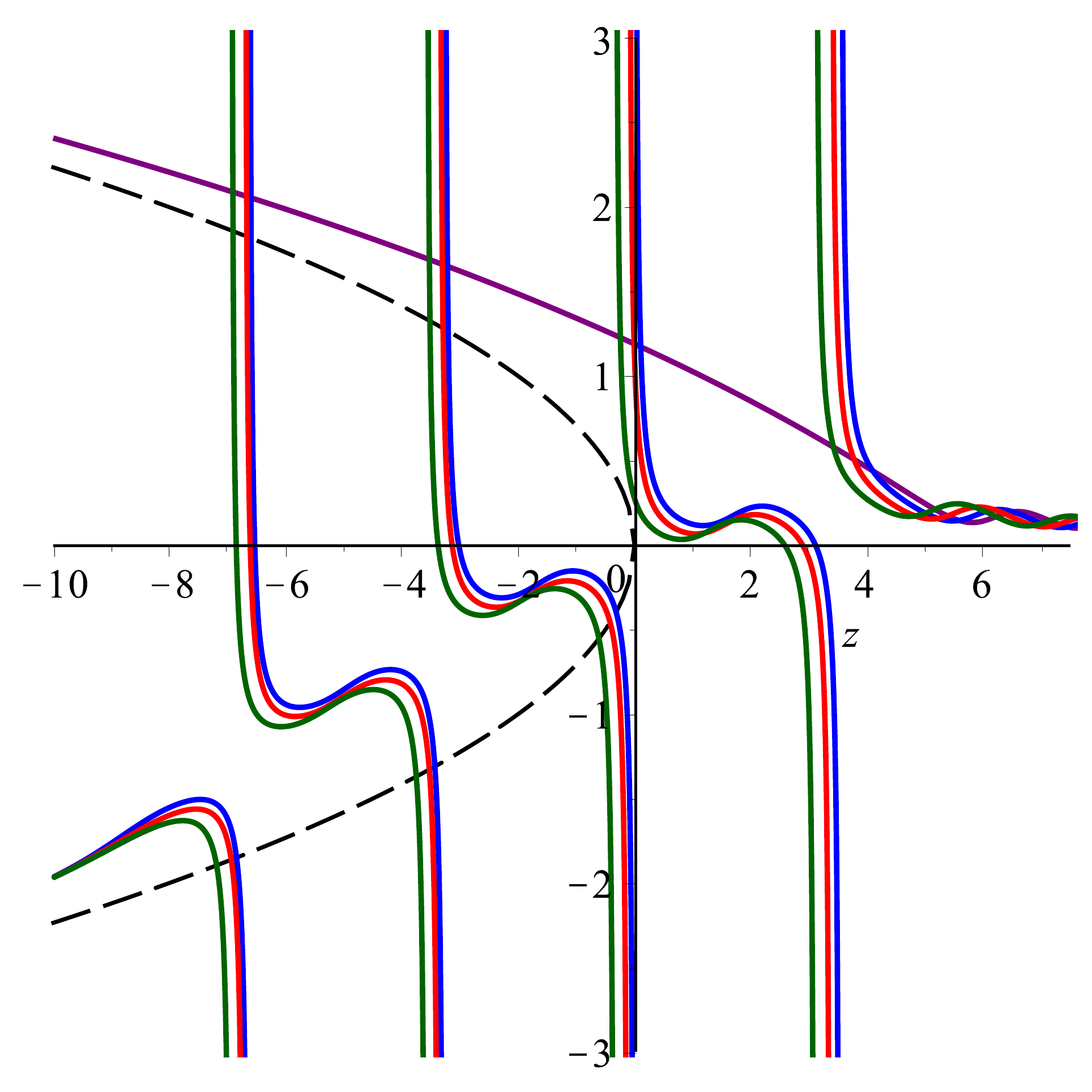}\\
{n=4}
& {n=6}
& {n=8}\end{array}\]
\caption{\label{fig:SIIplots1}Plots of $\sigma_n(z;\th)/n$, with $\sigma_n(z;\th)$ given by (\ref{eq:PT.OP.SII}),  for $\th={0}$ [\purple{purple}], $\th={\tfrac16\pi}$ [\blue{blue}], $\th={\tfrac13\pi}$ [\red{red}], $\th={\tfrac12\pi}$ [\green{green}]; the dashed line is the parabola $2\sigma^2+z=0$.}}\end{figure}
\begin{table}\[\begin{array}{|c|c|}\hline
n & \sigma_n(z;\th)\\ \hline
1& \phantom{\dfrac{1}{2}}  -\Phi \phantom{\dfrac{1}{2}} \\
2&\dfrac{1}{2\Phi^2+z}\\[10pt]
3&\dfrac{2z\Phi^2-\Phi+z^2}{4\Phi^3+2z\Phi+1}\\[10pt]
4&\dfrac{48\Phi^3-8z^2\Phi^2+28z\Phi-4z^3+9}{z(8z\Phi^4-16\Phi^3+8z^2\Phi^2-8z\Phi+2z^3-3)}  +\dfrac{3}{z}\\[10pt]\hline
\end{array}\]
\caption{\label{tab:SII}The Airy solutions $\sigma_n(z;\th)$ (\ref{eq:PT.OP.SII}).}\end{table}

As for the Airy solutions of \PII\ (\ref{eq:PT.DE.PII}), the Airy solutions of \Ptf\ (\ref{eq:PT.DE.P34}) and \SII\ (\ref{eq:PT.DE.SII}) can be expressed in terms of the determinant (\ref{eq:PT.OP.Ok1}).
\begin{theorem}{Let $\tau_n(z;\th)$ be the determinant (\ref{eq:PT.OP.Ok1})
for $n\geq1$, with $\varphi(z;\th)$ given by (\ref{eq:PT.OP.eq23}) and $\tau_0(z;\th)=1$,
then 
\begin{align} \label{eq:PT.OP.P34} 
p_n(z;\th)&= -2\deriv[2]{}{z}\ln{\tau_n(z;\th)},\\
\label{eq:PT.OP.SII} 
\sigma_n(z;\th)&= \deriv{}{z}\ln{\tau_n(z;\th)},
\end{align} 
satisfy \Ptf\ (\ref{eq:PT.DE.P34}) and \SII\ (\ref{eq:PT.DE.SII}) with $\a=n-\tfrac12$, respectively. }\end{theorem}

Comparing (\ref{eq:PT.OP.PII}) and (\ref{eq:PT.OP.SII}), we see that 
\begin{equation} q_n(z;\th)=\sigma_{n-1}(z;\th)-\sigma_n(z;\th).\end{equation}

The first few Airy solutions of \Ptf\ (\ref{eq:PT.DE.P34}) and \SII\ (\ref{eq:PT.DE.SII}) are given in Tables \ref{tab:P34} and \ref{tab:SII}, respectively, and plots 
in Figures \ref{fig:P34plots1} and \ref{fig:SIIplots1}, respectively.
As was the case for \PII, these plots show that the asymptotic behaviour as $z\to-\infty$ of the Airy solutions is completely different when $\th=0$ compared to the case when $\th\not=0$. 
The solutions $p_n(z;0)$ and $\sigma_n(z;0)$, which depends only on $\Ai(t)$ and have the same poles, are \textit{tronqu\'{e}e} solutions, 
see Figure \ref{fig:P34poles}.
Further we see that the solutions $p_{2n}(z;0)$ and $\sigma_{2n}(z;0)$, for $n\in\Z$, are special in that they have no poles on the real axis, see also Figures \ref{fig:P34plots2} and \ref{fig:SIIplots2}. Additionally these solutions have oscillatory behaviour with algebraic decay as $z\to\infty$ as given in Theorem \ref{P34SIIasymp2}, see also Figures \ref{fig:P34plots3} and \ref{fig:SIIplots3}. 

\begin{figure}
\[\begin{array}{c@{\quad}c@{\quad}c}
\fig{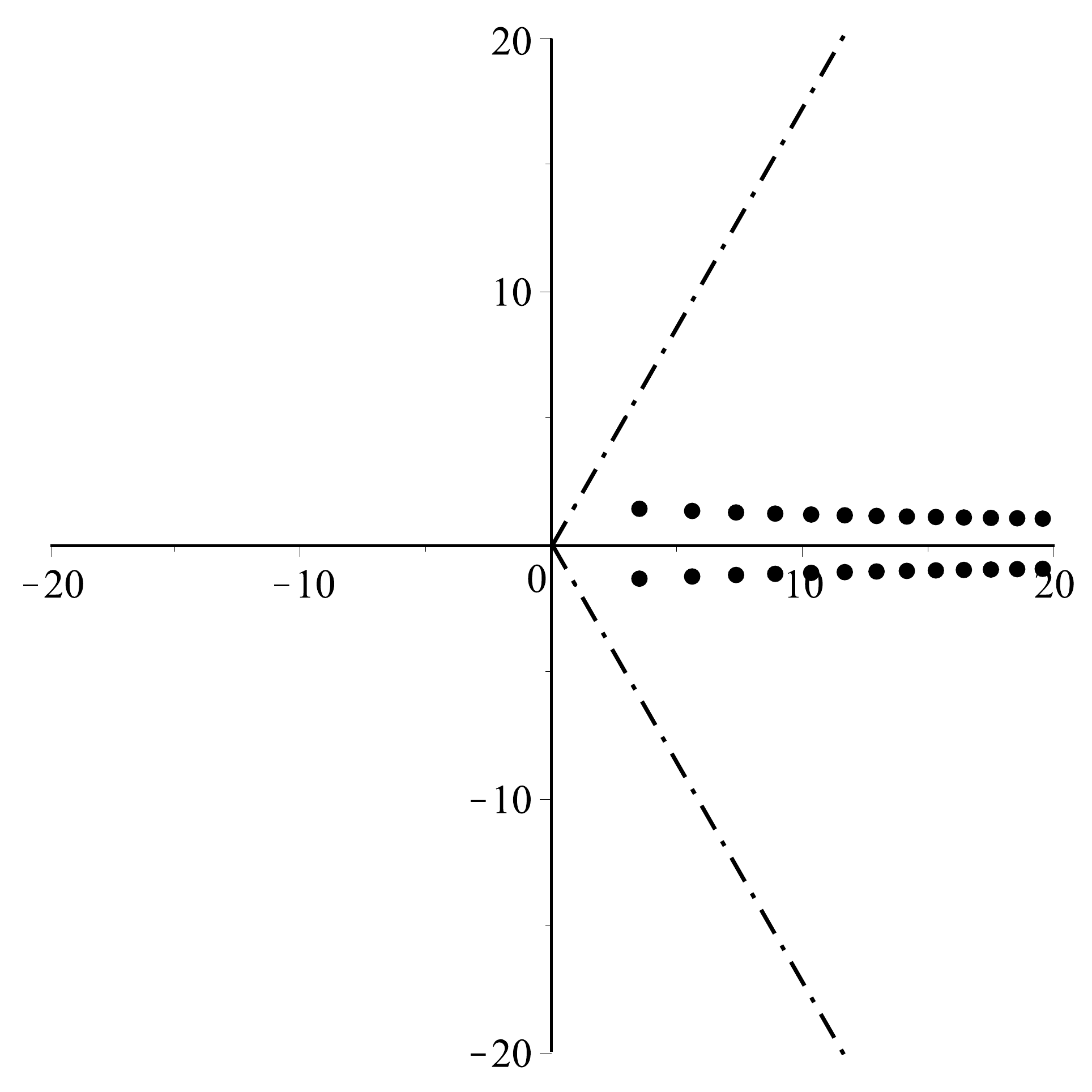}&\fig{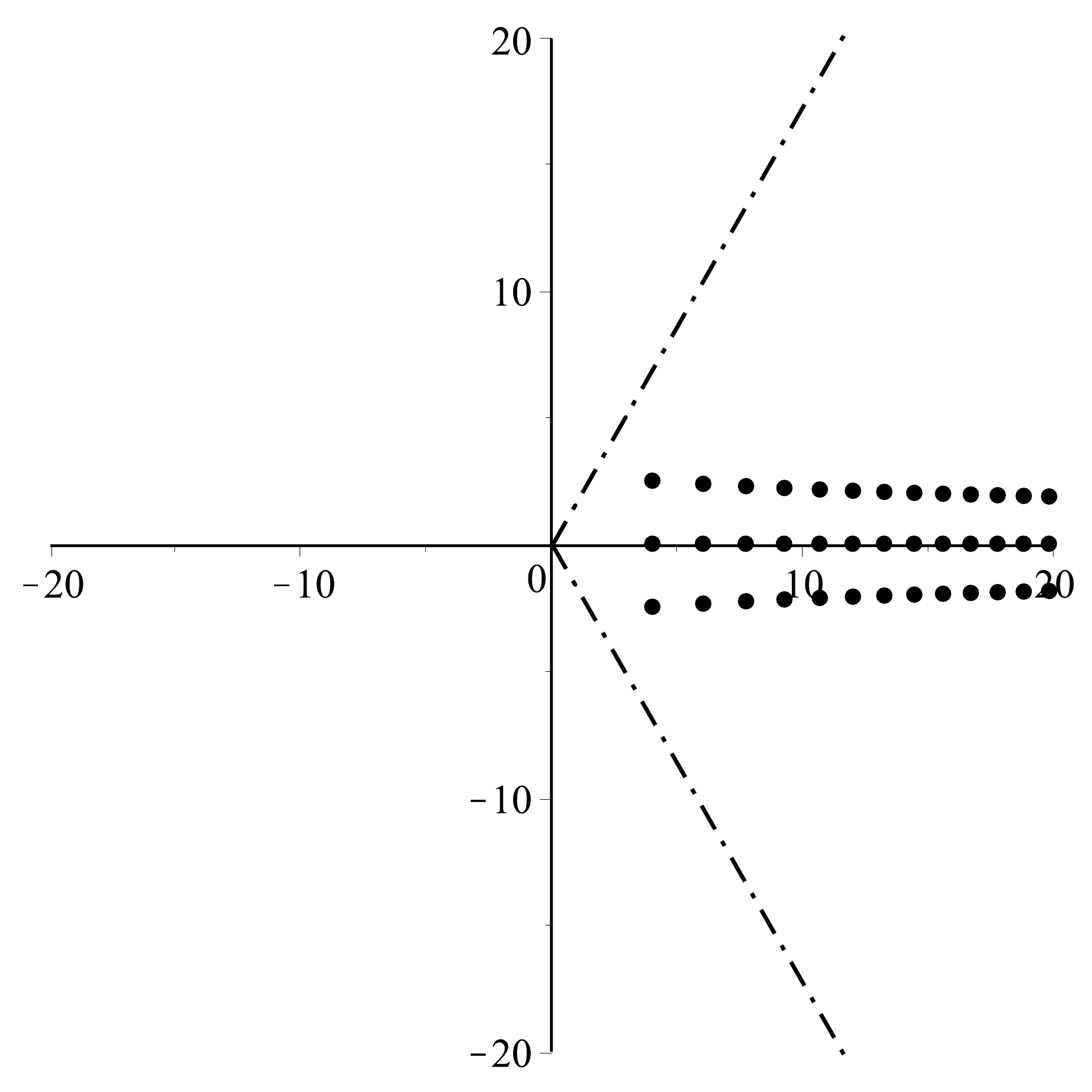}&\fig{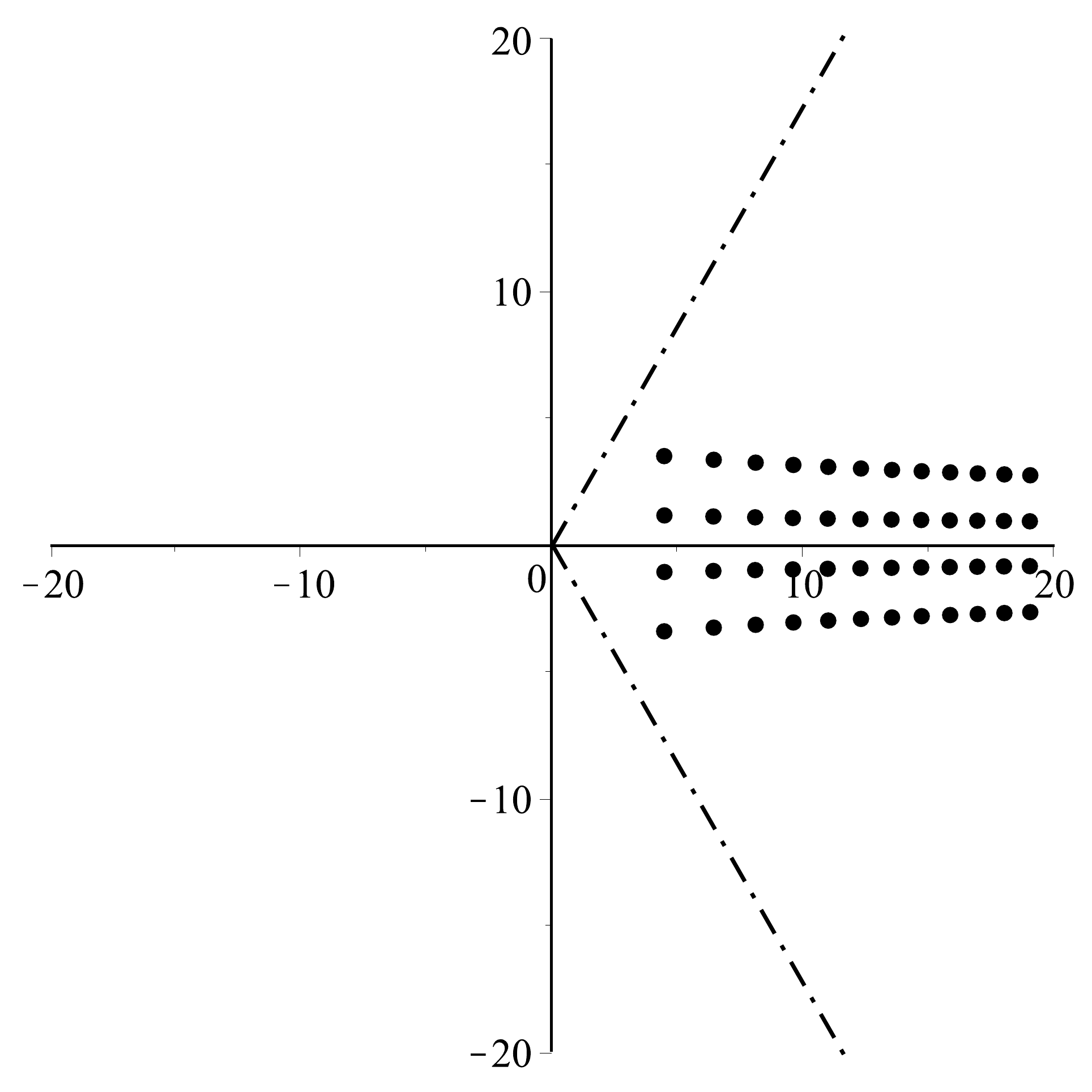}\\
n=2 & n=3 & n=4
\end{array}\]
\caption{\label{fig:P34poles}Plots of the poles of $p_n(z;0)$ and $\sigma_n(z;0)$ for $n=2,3,4$.}\end{figure}

{\begin{figure}{
\[\begin{array}{c@{\quad}c}
\figg{6cm,height=4.5cm}{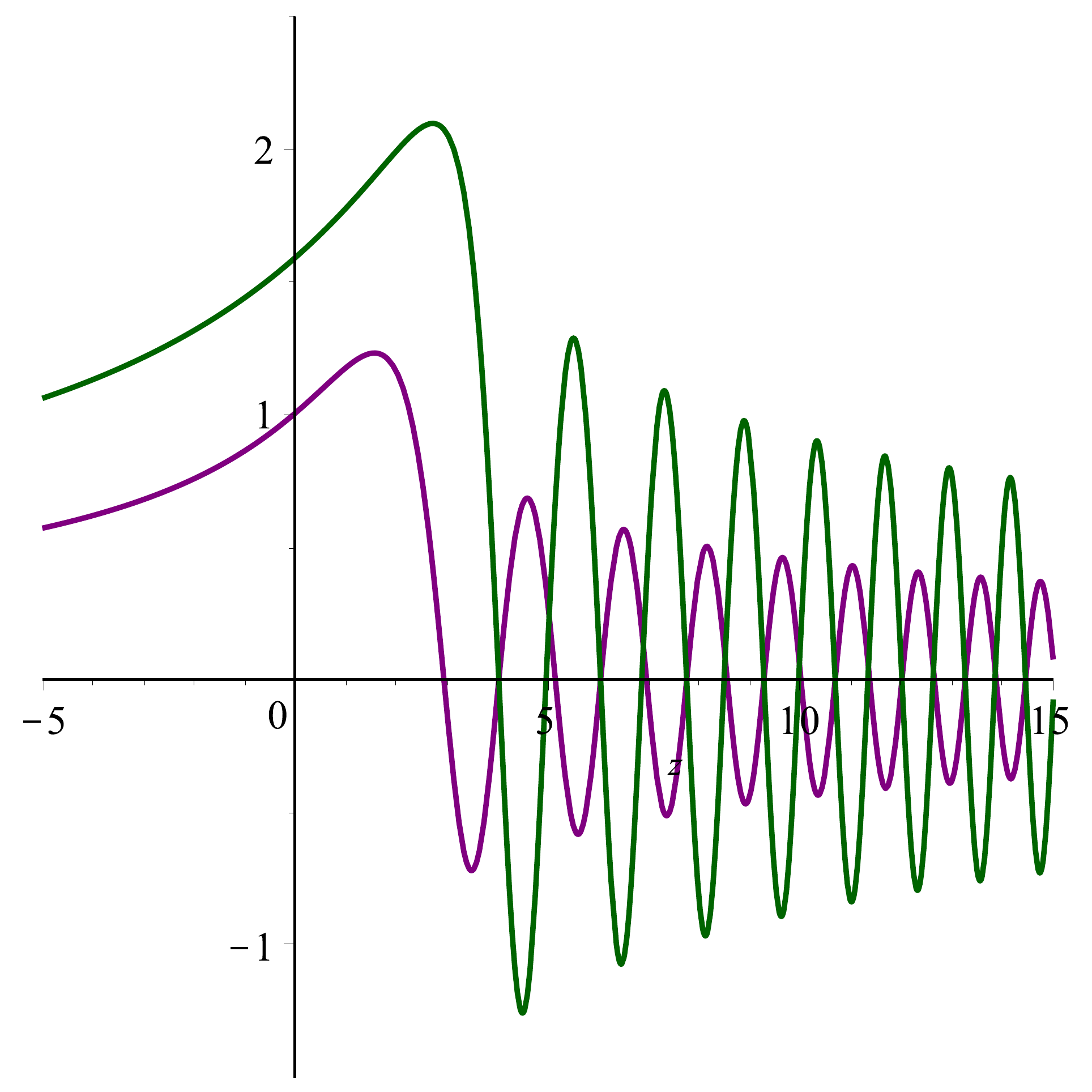}&\figg{6cm,height=4.5cm2,}{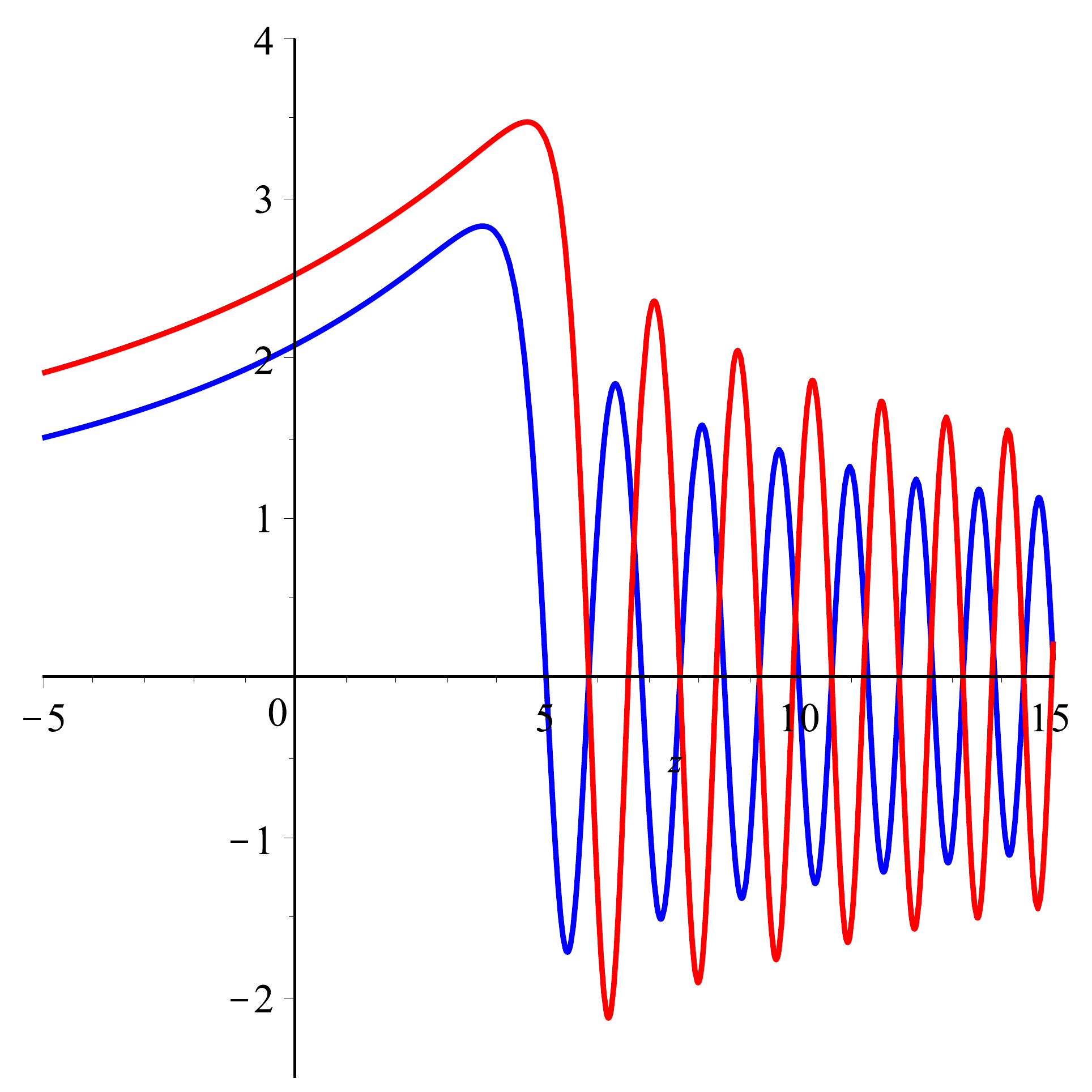}\\
n=2\ \mbox{\rm [\purple{purple}]}, n=4\ \mbox{\rm [\green{green}]} & n=6\ \mbox{\rm [\blue{blue}]}, n=8\ \mbox{\rm [\red{red}]}
\end{array}\]
\caption{\label{fig:P34plots2}Plots of $p_n(z;0)$ (\ref{eq:PT.OP.P34}), for $n={2},4,6,8$.} 
}\end{figure}}
{\begin{figure}{
\[\begin{array}{c}
\figg{8cm,height=6cm}{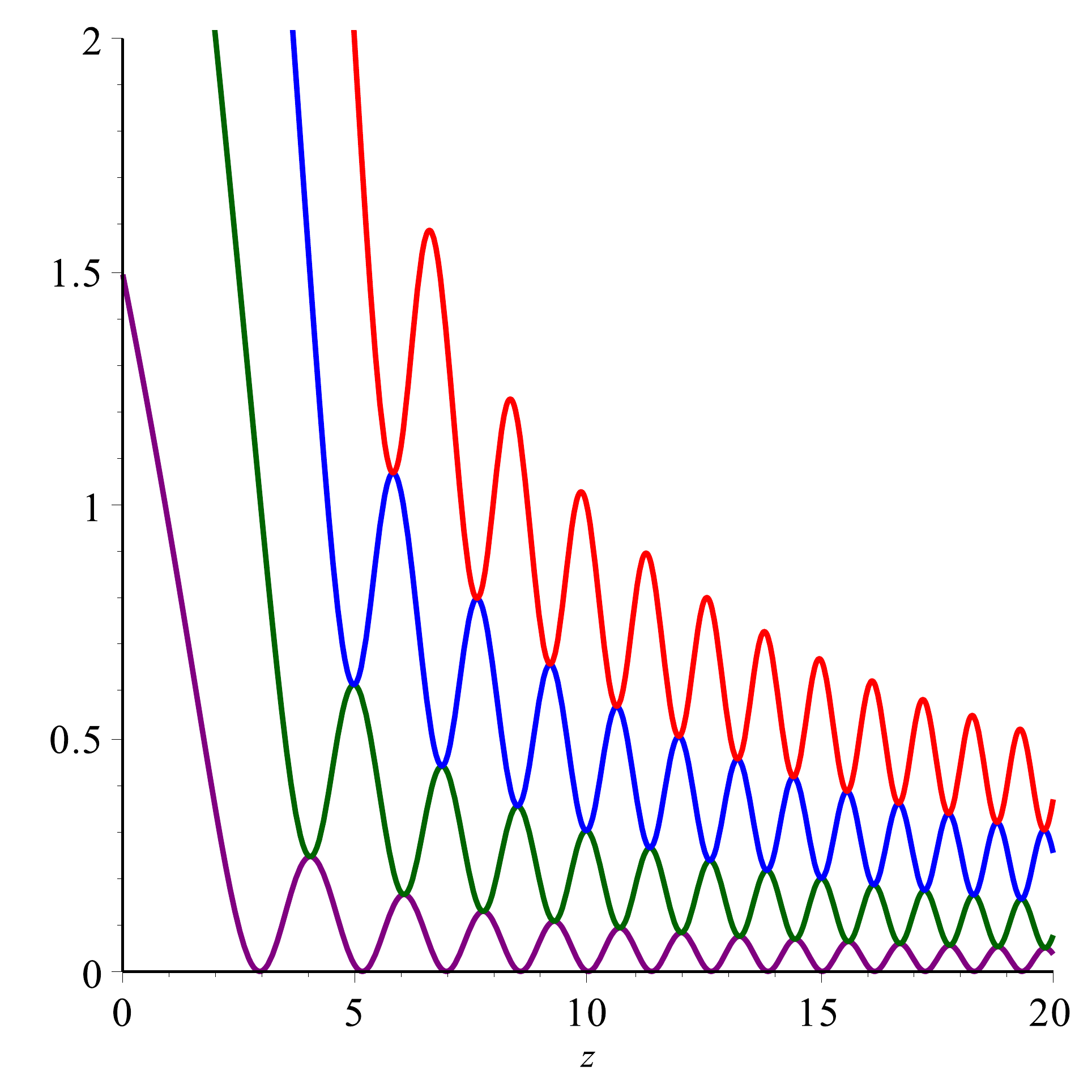}\\
n=2\ \mbox{\rm [\purple{purple}]}, n=4\ \mbox{\rm [\green{green}]}, n=6\ \mbox{\rm [\blue{blue}]}, n=8\ \mbox{\rm [\red{red}]}
\end{array}\] 
\caption{\label{fig:SIIplots2}Plots of $\sigma_n(z;0)$ (\ref{eq:PT.OP.SII}), for $n=2,4,6,8$.} 
}\end{figure}}

\begin{figure}
\[\begin{array}{c@{\quad}c@{\quad}c}
\fig{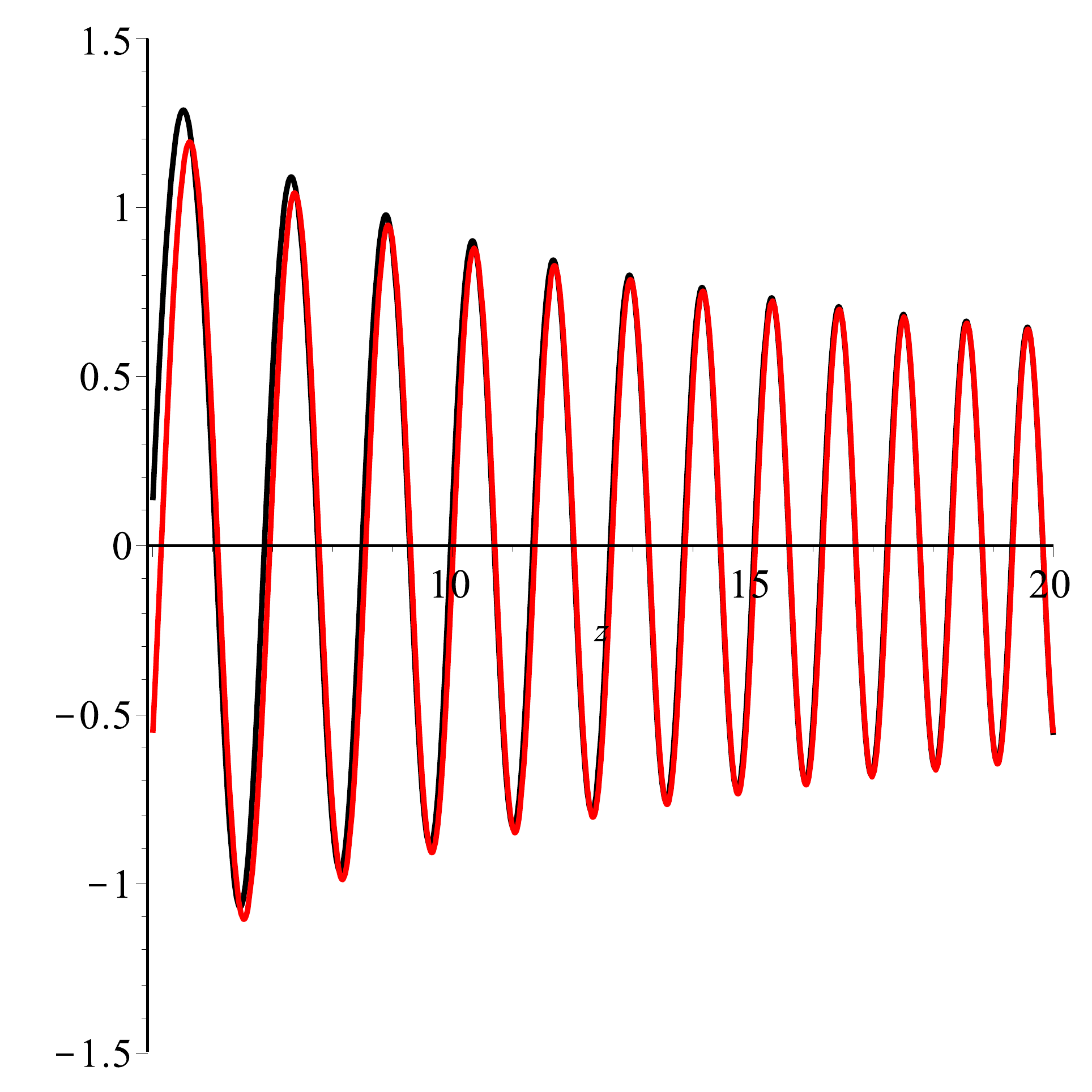}&\fig{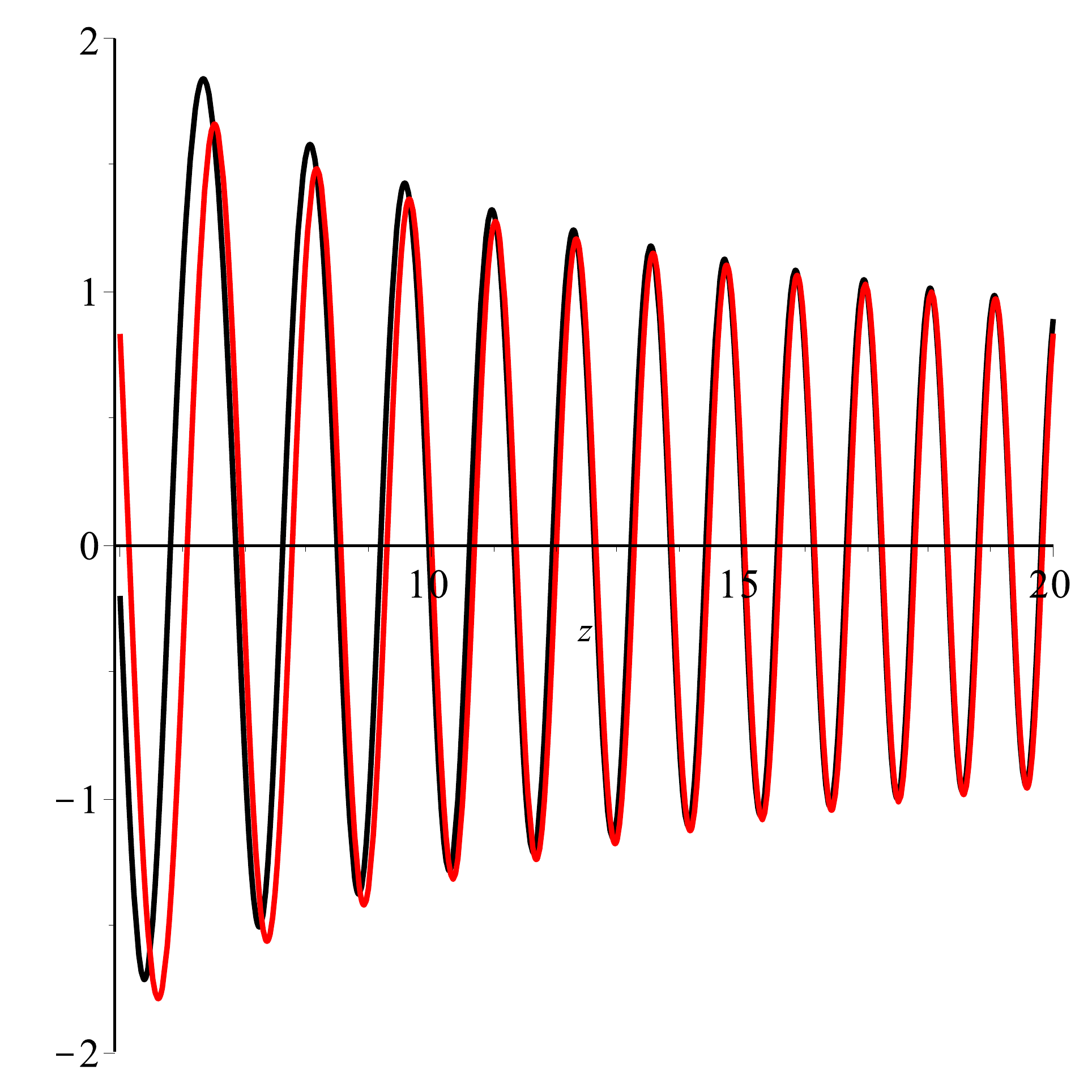}&\fig{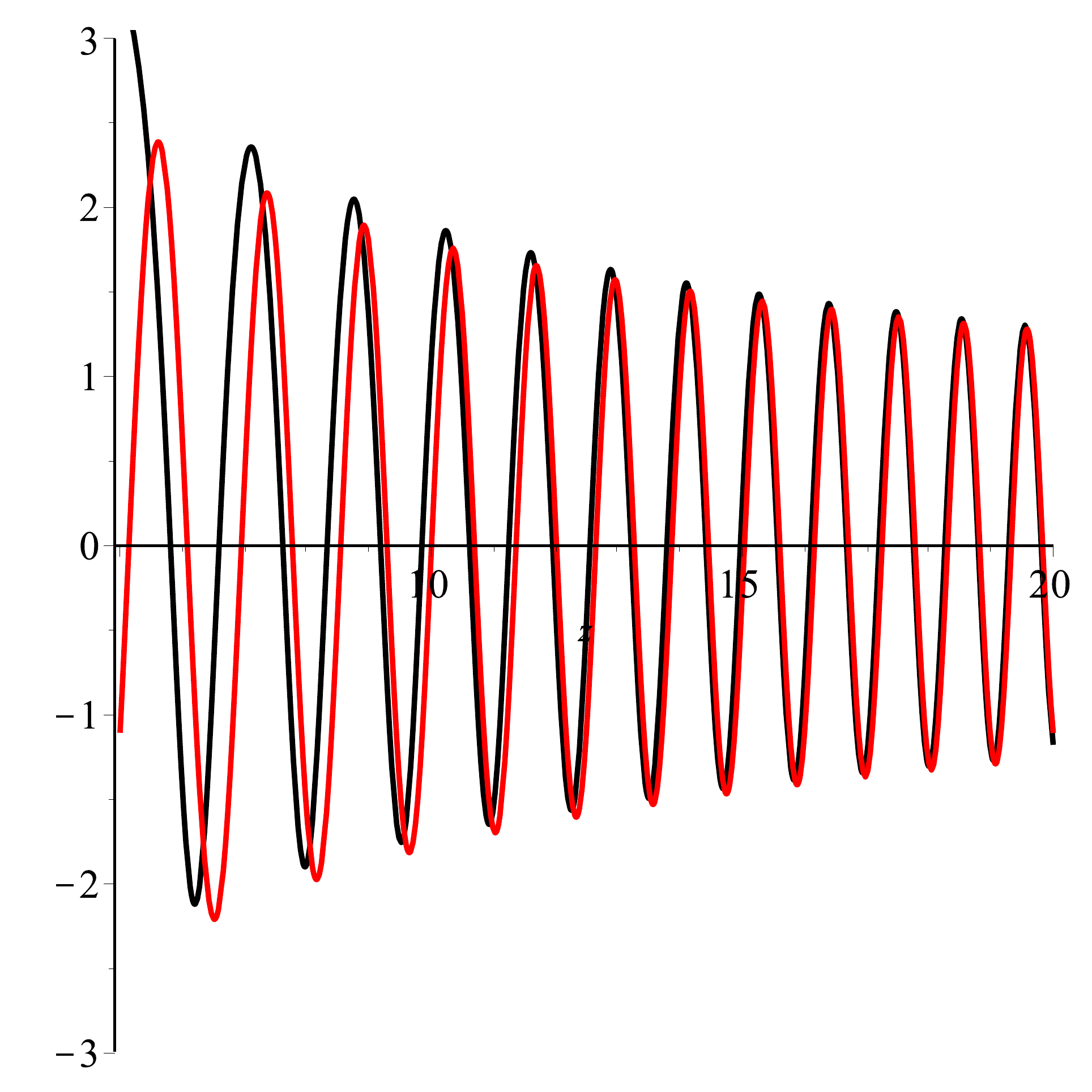}\\
n=4 & n=6 & n=8
\end{array}\]
\caption{\label{fig:P34plots3}Plots of $p_n(z;0)$ (\ref{eq:PT.OP.P34}) [\blue{blue}] and the leading term in the asymptotic expansion (\ref{asympn2}) [\red{red}] for $n=4,6,8$.}\end{figure}

\begin{figure}
\[\begin{array}{c@{\quad}c@{\quad}c}
\fig{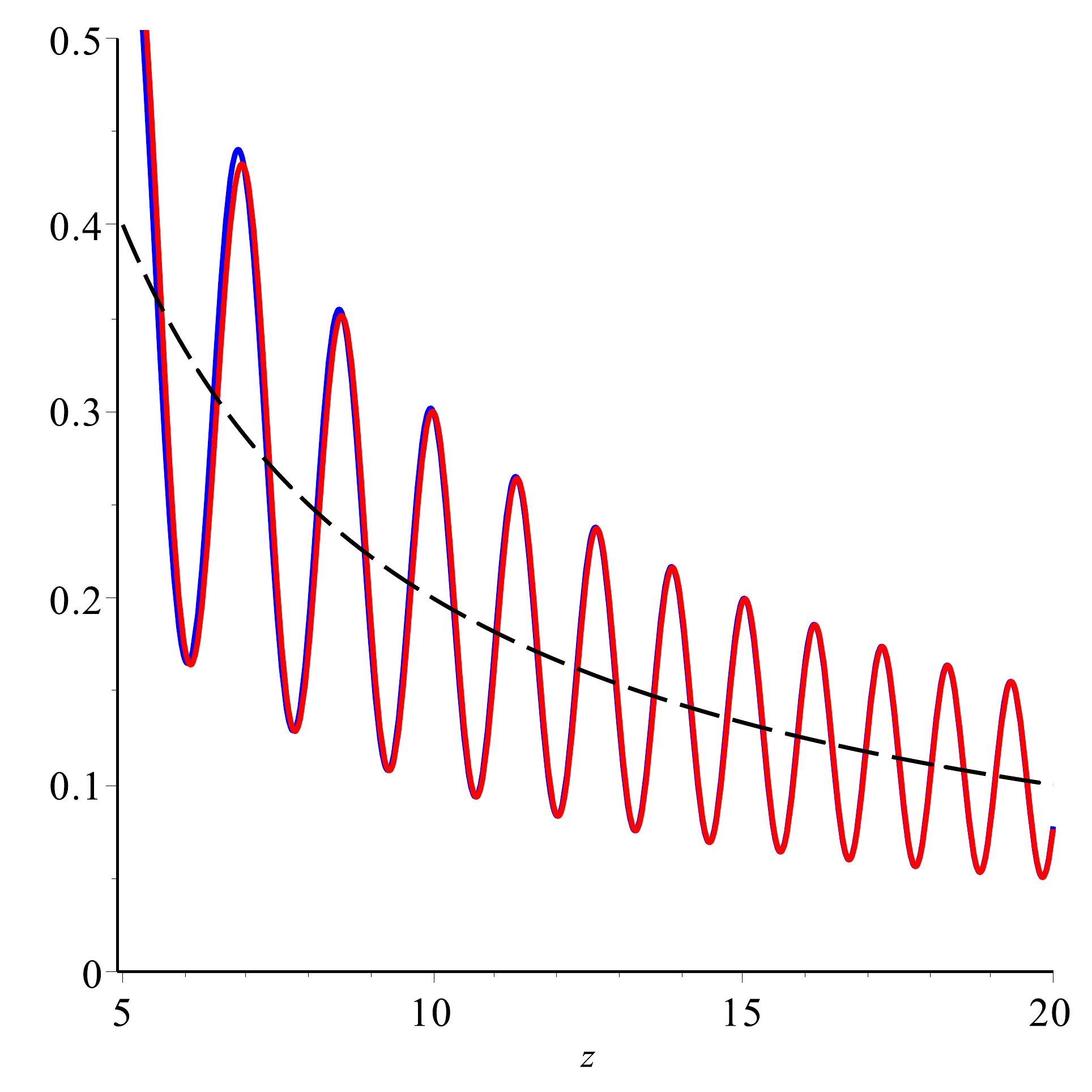}&\fig{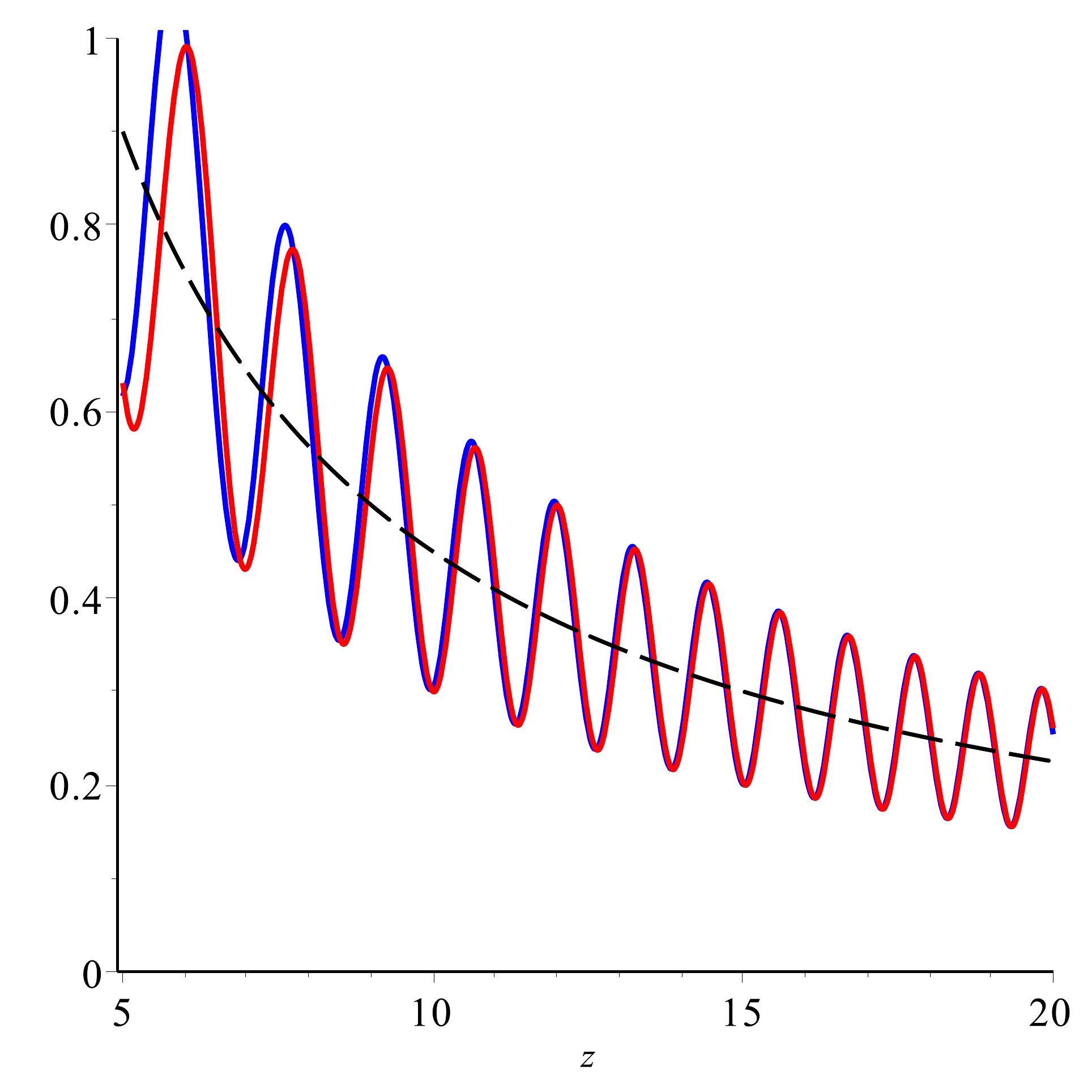}&\fig{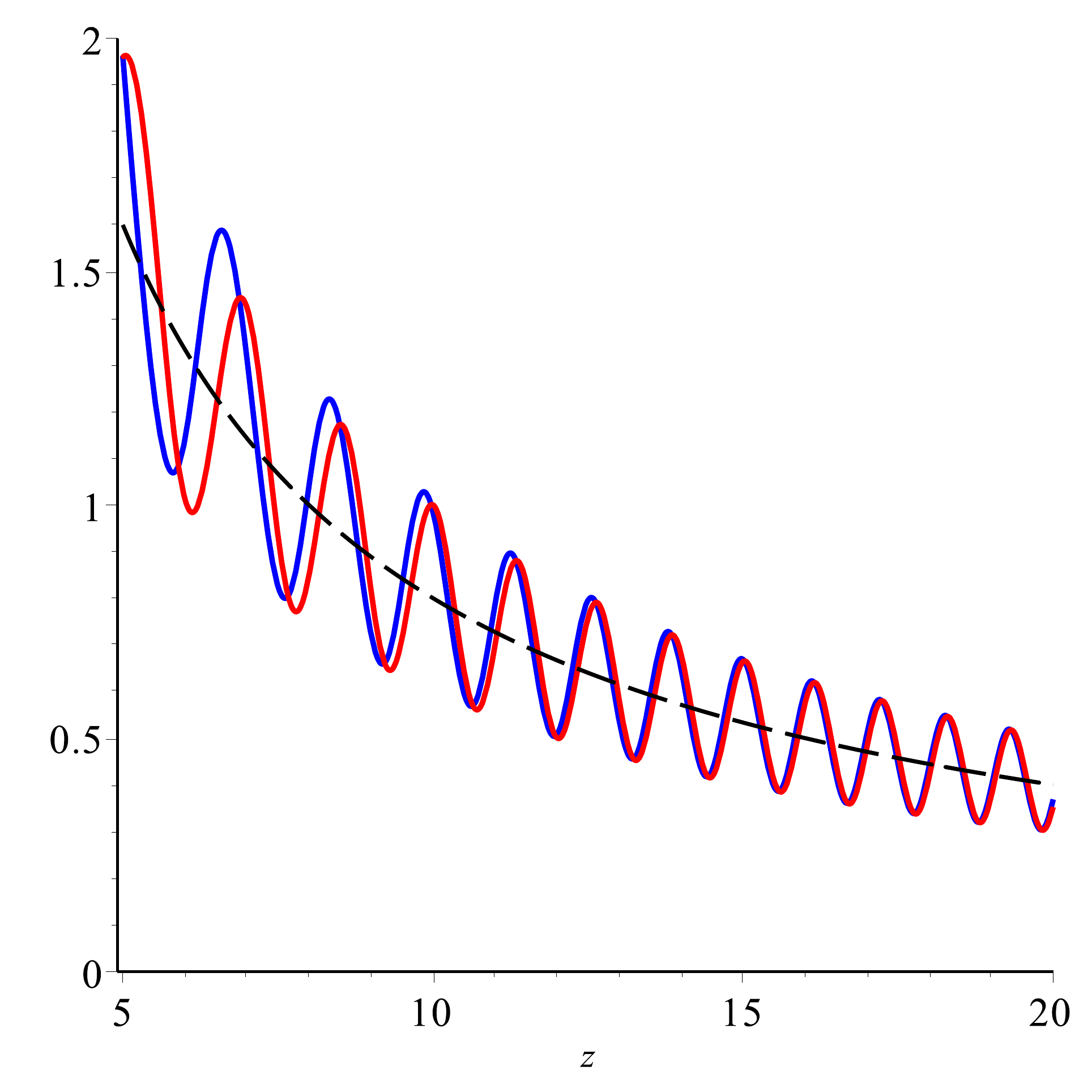}\\
n=4 & n=6 & n=8
\end{array}\]
\caption{\label{fig:SIIplots3}Plots of $\sigma_n(z;0)$ (\ref{eq:PT.OP.SII}) [\blue{blue}] and the leading term in the asymptotic expansion (\ref{asymsigman2}) [\red{red}] for $n=4,6,8$; the dashed line is $n^2/(8z)$.}\end{figure}

The asymptotic behaviour of the Airy solutions $p_n(z;\th)$ and $\sigma_n(z;\th)$  as $z\to-\infty$ is given in the following theorem.

\begin{theorem}{\label{P34SIIasymp1}Let $p_n(z;\th)$ be defined by (\ref{eq:PT.OP.P34}) and $\sigma_n(z;\th)$ by (\ref{eq:PT.OP.SII}), 
then as $z\to-\infty$, 
\begin{align}\label{asympn1}
p_n(z;\th)&=\begin{cases}
\ds \frac{n}{\sqrt{2}\,(-z)^{1/2}}-\frac {n^2}{2z^2}+{\frac {5n(4n^2+1) }{16\sqrt {2}\,(-z)^{7/2}}}
+\O\left(z^{-5}\right), & \mbox{if}\quad  \th=0,\\
\ds -\frac{n}{\sqrt{2}\,(-z)^{1/2}}-\frac{n^2}{2z^2}-{\frac {5n(4n^2+1)}{16\sqrt {2}\,(-z)^{7/2}}}
+\O\left(z^{-5}\right), & \mbox{if}\quad  \th\not=0,
\end{cases}\\ \label{asymsigman1}
\sigma_n(z;\th)&=\begin{cases}
\ds\frac{n(-z)^{1/2}}{\sqrt {2}}-\frac {n^2}{4z}-{\frac {n(4n^2+1) }{16\sqrt {2}\,(-z)^{5/2}}} 
 +\O\left(z^{-4}\right), & \mbox{if}\quad  \th=0,\\
\ds -\frac{n(-z)^{1/2}}{\sqrt {2}}-\frac {n^2}{4z}+{\frac {n(4n^2+1) }{16\sqrt {2}\,(-z)^{5/2}}} 
+\O\left(z^{-4}\right), & \mbox{if}\quad  \th\not=0.
\end{cases}\end{align}
}\end{theorem}

\begin{proof}The asymptotic expansions (\ref{asympn1}) and (\ref{asymsigman1}) are easily derived using the expansion for $q_n(z;\th)$ given in Theorem \ref{thm33} since if $q$ is solution of \PII\ (\ref{eq:PT.DE.PII}), then
\[ \begin{split} p &= \deriv{q}{z}+q^2+\tfrac12z,\\
\sigma &
= \tfrac12\left(\deriv{q}{z}\right)^{\!2}-\tfrac12 q^4-\tfrac12zq^2-(\a+\tfrac12)q-\tfrac18z^2,\end{split}\]
are solutions of \Ptf\ (\ref{eq:PT.DE.P34}) and \SII\ (\ref{eq:PT.DE.SII}), respectively
\end{proof}

We remark that Its, Kuijlaars, and \"{O}stensson \cite{refIKO08} noted that $p_2(z;0)$ is special among all solutions $p_2(z;\th)$ in its behavior as $z\to-\infty$. 

The asymptotic behaviour of the solutions $p_n(z;\th)$ and $\sigma_n(z;\th)$ as $z\to\infty$, for $n\in2\Z$, is given in the following theorem.

\begin{theorem}{\label{P34SIIasymp2}Let $p_n(z;\th)$ be defined by (\ref{eq:PT.OP.P34}) and $\sigma_n(z;\th)$ by (\ref{eq:PT.OP.SII})
and $n\in2\Z$, then as $z\to\infty$, 
\begin{align}\label{asympn2}
p_{n}(z;0)&=\tfrac12n\sqrt{2} \,z^{-1/2}\cos
\left(\tfrac43\sqrt{2}\,z^{3/2}-\tfrac12n\pi\right)+o\big(z^{-1/2}\big),\\
\sigma_{n}(z;0)&=\frac{n}{8z}\left\{n-2 
\sin\left(\tfrac43\sqrt{2}\,z^{3/2}-\tfrac12n\pi\right) \right\}+o\big(z^{-1}\big).\label{asymsigman2}
\end{align}}\end{theorem}

\begin{proof}
Using the known asymptotic expansion of the Airy function as $t\to-\infty$, i.e.
\[\Ai(t)=\pi^{-1/2}(-t)^{-1/4}\sin\big\{\tfrac23(-t)^{3/2}+\tfrac14\pi\big\}+o\big((-t)^{-1/4}\big)\]
\cite[\S9.7(ii)]{refDLMF} it is straightforward to verify the asymptotic expansions (\ref{asympn2}) and (\ref{asymsigman2}) for small values of $n$.
\end{proof}

Plots of $p_n(z;0)$ (\ref{eq:PT.OP.P34}) and the leading term in the asymptotic expansion (\ref{asympn2}) for $n=4,6,8$ are given in Figure \ref{fig:P34plots3}. Analogous plots of $\sigma_n(z;0)$ (\ref{eq:PT.OP.SII}) and the leading term in the asymptotic expansion (\ref{asymsigman2}) for $n=4,6,8$ are given in Figure \ref{fig:SIIplots3}.

\begin{remark}{\rm
Its, Kuijlaars, and \"{O}stensson \cite{refIKO08} discussed solutions of the equation
\begin{equation} u_{\b}\deriv[2]{u_{\b}}{t}=\tfrac12\left(\deriv{u_{\b}}{t}\right)^{\!2}+4u_{\b}^3+2tu_{\b}^2-2\b^2,\label{eq:IKO}\end{equation}
where $\b$ is a constant,
which is equivalent to \Ptf\ (\ref{eq:PT.DE.P34}) through the transformation $p(z)=2^{1/3}u_{\b}(t)$, with $t=-2^{-1/3}z$, and
$\b=\tfrac12\a+\tfrac14$ in their study of the double scaling limit of unitary random
matrix ensembles of the form 
$$Z^{-1}_{n,N} |\det M|^{2\b}\exp\{-N\Tr V (M)\} dM,$$with $\b > -\tfrac12$ and $V$ real analytic. In particular,
\begin{equation} u_1(t)=-\deriv[2]{}{t}\,
\mathcal{W}\big(\Ai(t),\Ai'(t)\big), \label{sol:IKOu1}\end{equation}
with $\mathcal{W}(\varphi_1,\varphi_2)$ the Wronskian,
which is equivalent to the solution $p_2(z;0)$ of \Ptf\ (\ref{eq:PT.DE.P34}), and noted that on the positive real axis, $u_1(t)$ has an infinite number of zeros, which are the zeros of the Airy function $\Ai(t)$, and an infinite number of additional zeros that interlace with the zeros of $\Ai(t)$, see Figure \ref{fig:IKOfig6}. }\end{remark}
\begin{figure}
\[\figg{8cm,height=6cm}{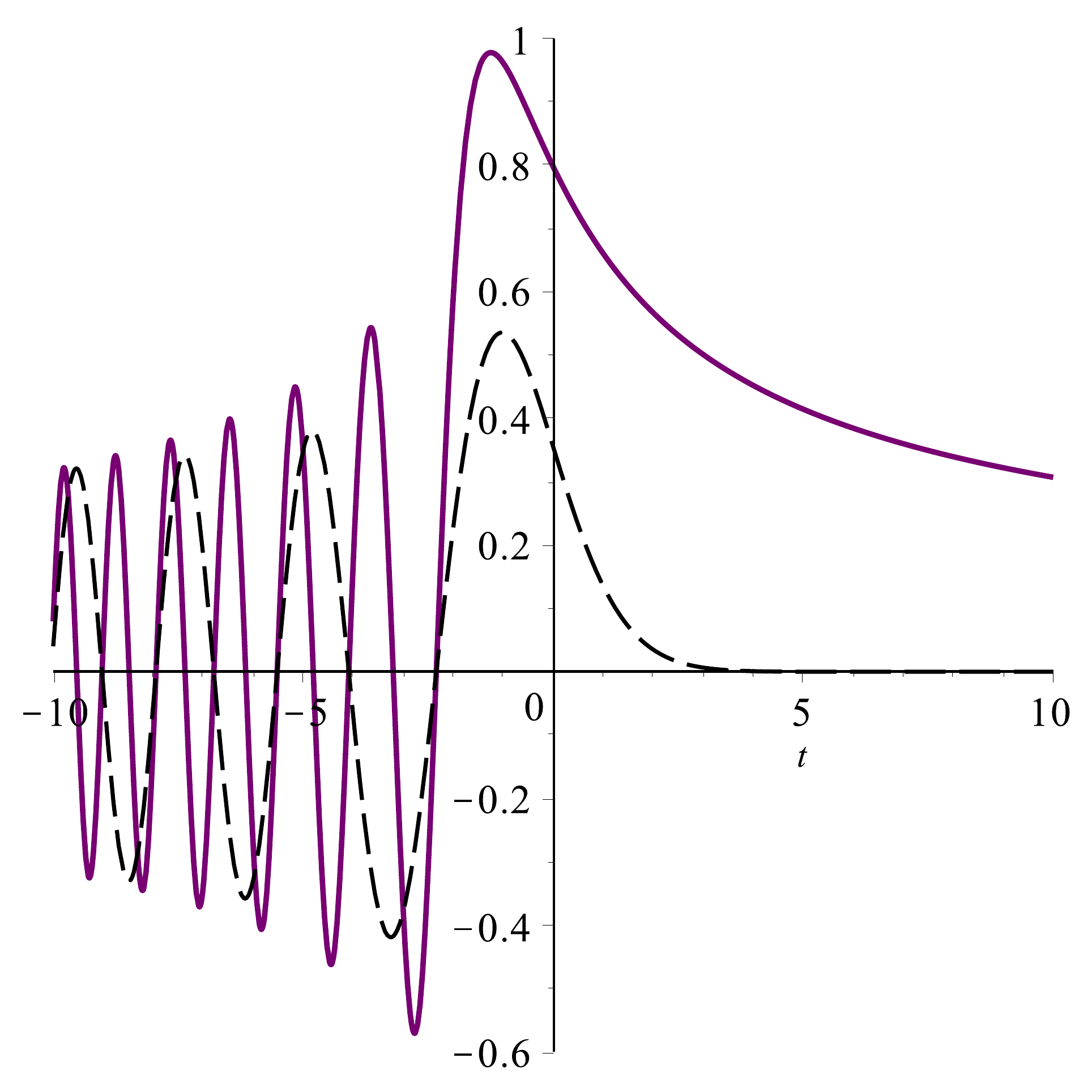}\]
\caption{\label{fig:IKOfig6}The solution $u_1(t)$ (\ref{sol:IKOu1}) of equation (\ref{eq:IKO}) [\purple{purple}] and the Airy function $\Ai(t)$ [dashed line].}
\end{figure}

\begin{remark}{\rm
Its, Kuijlaars, and \"{O}stensson \cite[Theorem 1.2]{refIKO09} prove that there are solutions 
\comment{$p(z;\a)$ of \Ptf\ (\ref{eq:PT.DE.P34}), for general $\a$, such that as $z\to-\infty$
\begin{equation} p(z;\a) = \ds{\frac { 2\a+1}{2\sqrt {2}\,(-z)^{1/2}}} + \O\big(z^{-2}\big),\label{asymIKO1}\end{equation}
and as $z\to\infty$
\begin{equation} p(z;\a) =\ds{\frac {2\a+1}{2\sqrt {2z}}}\cos \left\{\tfrac43\sqrt{2}\,{z}^{3/2}-\tfrac14(2\a+1)\pi \right\}+ \O\big(z^{-2}\big).\label{asymIKO2}\end{equation} }%
$u_\b(t)$ of (\ref{eq:IKO}) such that as $t\to\infty$
\begin{align}
u_\b(t)&= {\b}t^{-1/2} + \O\big(t^{-2}\big),&&\mbox{as $t\to\infty$}.\label{asymIKO1} \\
u_\b(t)&= {\b}(-t)^{-1/2}\cos \left\{\tfrac43(-t)^{3/2}-\b\pi \right\}+ \O\big(t^{-2}\big),
&&\mbox{as $t\to-\infty$}.\label{asymIKO2}\end{align}
Letting $\b=1$ in (\ref{asymIKO1}) and (\ref{asymIKO2}) shows that they are in agreement with (\ref{asympn1}) and (\ref{asympn2}), for $n=2$. Further Its, Kuijlaars, and \"{O}stensson conclude that solutions of (\ref{eq:IKO}) with asymptotic behaviour (\ref{asymIKO1}) are \textit{tronqu\'{e}e} solutions, i.e.\ they have no poles in a sector of the complex plane.
}\end{remark}

The plots in Figures \ref{fig:P34plots1} and \ref{fig:SIIplots1} suggest the following conjecture, which is analogous to Conjecture \ref{thm34}.

\begin{conjecture}{If $p_n(z;0)$ is defined by (\ref{eq:PT.OP.P34}) and $\sigma_n(z;0)$ by (\ref{eq:PT.OP.SII}), then for $z<0$ and $n\geq1$, $p_n(z;0)$ is a monotonically decreasing function and $\sigma_n(z;0)$ is a monotonically increasing function as $z$ decreases. Further,  for $z<0$
\[p_{n+1}(z;0)<p_n(z;0),\qquad 
\sigma_{n}(z;0)<\sigma_{n+1}(z;0).\]}\end{conjecture}

\section{Conclusion}
In this paper we have studied the Airy solutions of \PII\ (\ref{eq:PT.DE.PII}), \Ptf\ (\ref{eq:PT.DE.P34}), and \SII\ (\ref{eq:PT.DE.PII}). These show that when the solutions depend only on the Airy function $\Ai(t)$, with $t=-2^{-1/3}z$, have a completely different asymptotic behaviour as $z\to-\infty$ compared to solutions which involve the Airy function $\Bi(t)$.  The special solutions which depend only on $\Ai(t)$ are tronqu\'{e}e solutions, i.e.\ they have no poles in a sector of the complex plane. Further for both \Ptf\ (\ref{eq:PT.DE.P34}) and \SII\ (\ref{eq:PT.DE.PII}), it is shown that amongst these tronqu\'{e}e solutions there is a family of solutions which have no poles on the real axis.  Such solutions of \p\ equations often have important physical applications. For example, the well-known Hastings-McLeod solution of \PII\ (\ref{eq:PT.DE.PII}) with $\a=0$ \cite{refHMcL}, which arises in numerous applications, is a tronqu\'{e}e solution with no poles on the real axis, cf.~\cite{refMXZ,refNovok09,refNovok12}.

The special solutions $q_n(z;0)$ and $p_n(z;0)$ of \PII\ and \Ptf\ arise in recent study by Clarkson, Loureiro, and Van Assche \cite{refCLvA} of the discrete system
\begin{align*} 
&a_n+a_{n+1} = b_{n}^2 -t , \\
 &a_n(b_{n}+b_{n-1}) = n, 
\end{align*}
which is equivalent to alt-\dPI\ (\ref{eq:PT.BT.P23}) 
and arises in the study of orthogonal polynomials with respect to an exponential cubic weight 
\cite{refBD,refFilWVALun,refMagnus95}.

\section*{Acknowledgments}
I would like to thank Mark Ablowitz, Andrew Bassom, Andrew Hone, Alexander Its, Kerstin Jordaan, Nalini Joshi, Ana Loureiro, Elizabeth Mansfield, Marta Mazzocco, Bryce McLeod, and Walter Van Assche for their helpful comments and illuminating discussions.

\def\ams{American Mathematical Society}
\def\AAM{Acta Appl. Math.}
\def\ARMA{Arch. Rat. Mech. Anal.}
\def\ampa{Ann. Mat. Pura Appl. (IV)}
\def\cmp{Commun. Math. Phys.}
\def\CMP{Commun. Math. Phys.}
\def\cpam{Commun. Pure Appl. Math.}
\def\CPAM{Commun. Pure Appl. Math.}
\def\CQG{Classical Quantum Grav.}
\def\CSF{Chaos, Solitons \&\ Fractals}
\def\DE{Diff. Eqns.}
\def\DU{Diff. Urav.}
\def\ejam{Europ. J. Appl. Math.}
\def\EJAM{Europ. J. Appl. Math.}
\def\funk{Funkcial. Ekvac.}
\def\FUNK{Funkcial. Ekvac.}
\def\IP{Inverse Problems}
\def\JDE{J. Diff. Eqns.}
\def\JMAA{J. Math. Anal. Appl.}
\def\JMMM{J. Magn. Magn. Mater.}
\def\JMP{J. Math. Phys.}
\def\jmp{J. Math. Phys}
\def\JNMP{J. Nonl. Math. Phys.}
\def\jpa{J. Phys. A}
\def\JPA{J. Phys. A}
\def\JPSJ{J. Phys. Soc. Japan}
\def\LMP{Lett. Math. Phys.}
\def\NL{Nonlinearity}
\def\NMJ{Nagoya Math. J.}
\def\pl{Phys. Lett.}
\def\PL{Phys. Lett.}
\def\PR{Phys. Rev.}
\def\PRL{Phys. Rev. Lett.}
\def\prsl{Proc. R. Soc. Lond. A}
\def\PRSL{Proc. R. Soc. Lond. A}
\def\SAM{Stud. Appl. Math.}
\def\sam{Stud. Appl. Math.}
\def\JCAM{J. Comput. Appl. Math.}

\def\OUP{O.U.P.}
\def\CUP{C.U.P.}
\def\AMS{American Mathematical Society}

\def\refpp#1#2#3#4#5{\vspace{-0.25cm}
\bibitem{#1} {\frenchspacing#2}, \textrm{#3}, #4 (#5).}

\def\refjl#1#2#3#4#5#6#7{\vspace{-0.25cm}
\bibitem{#1} {\frenchspacing#2}, {\rm#6}, 
\textit{\frenchspacing#3}, \textbf{#4} (#7) #5.}

\def\refbk#1#2#3#4#5{\vspace{-0.25cm}
\bibitem{#1} {\frenchspacing#2}, \textit{#3}, #4, #5.}

\def\refcf#1#2#3#4#5#6{\vspace{-0.25cm}
\bibitem{#1} {\frenchspacing#2}, \textrm{#3},
in \textit{#4}, {\frenchspacing#5}, #6.}

\def\and{\mbox{\rm and}\ }


\end{document}